\documentclass[english]{article}
\usepackage[T1]{fontenc}
\usepackage[latin9]{inputenc}
\usepackage[dvipsnames,svgnames,x11names,hyperref]{xcolor}
\usepackage{amsmath,amsthm} 
\usepackage{graphicx}
\usepackage{amssymb}
\usepackage{enumitem}
\usepackage{natbib} 
\usepackage{bussproofs}
\usepackage[letterpaper]{geometry} 
\usepackage{thmtools,thm-restate}
\usepackage{tikz}
\usepackage{pgf} 
\usetikzlibrary{patterns,automata,arrows,shapes,snakes,topaths,trees,backgrounds,positioning,through,calc}
\usepackage{setspace}
\usepackage{multicol}
\usepackage{graphicx}
\usepackage{stmaryrd}  
\usepackage{comment}
\usepackage{hyperref}
\hypersetup{colorlinks=true,citecolor=ForestGreen,linkcolor=ForestGreen}  

\usepackage{mathrsfs}

\newlist{myitemize}{itemize}{8}

\usepackage{natbib}

\theoremstyle{definition}
\newtheorem{theorem}{Theorem}[section]

\newtheorem{proposition}[theorem]{Proposition}
\newtheorem{example}[theorem]{Example}
\newtheorem{definition}[theorem]{Definition}
\newtheorem{lemma}[theorem]{Lemma}
  
\newtheorem{remark}[theorem]{Remark} 
 
\newtheorem{corollary}[theorem]{Corollary}

\usepackage[page,toc,titletoc,title]{appendix}

\setcitestyle{round}


\onehalfspace

\title{Escaping Arrow's Theorem: \\ The Advantage-Standard Model}

 \author{Wesley H. Holliday$^\dagger$ and Mikayla Kelley$^\ddagger$ \\ {\small $\dagger$ University of California, Berkeley and $\ddagger$ University of Chicago} }
 
 \date{{\normalsize Forthcoming in \textit{Theory and Decision}}.}
 
\begin{document}

\maketitle

\begin{abstract} There is an extensive literature in social choice theory studying the consequences of weakening the assumptions of Arrow's Impossibility Theorem. Much of this literature suggests that there is no escape from Arrow-style impossibility theorems, while remaining in an ordinal preference setting, unless one drastically violates the Independence of Irrelevant Alternatives (IIA). In this paper, we present a more positive outlook. We propose a model of comparing candidates in elections, which we call the Advantage-Standard (AS) model. The requirement that a collective choice rule (CCR) be representable by the AS model captures a key insight of IIA but is weaker than IIA; yet it is stronger than what is known in the literature as weak IIA (two profiles alike on $x,y$ cannot have opposite strict social preferences on $x$ and $y$). In addition to motivating violations of IIA, the AS model makes intelligible violations of another Arrovian assumption: the negative transitivity of the strict social preference relation $P$. While  previous literature shows that only weakening IIA to weak IIA or only weakening negative transitivity of $P$ to acyclicity still leads to impossibility theorems, we show that jointly weakening IIA to AS representability and weakening negative transitivity of $P$ leads to no such impossibility theorems. Indeed, we show that several appealing CCRs are AS representable, including even transitive~CCRs.\end{abstract}

\tableofcontents

\section{Introduction}

Arrow's Impossibility Theorem (\citealt{Arrow1951,Arrow1963}) states that for social choice problems involving at least three alternatives, any method of transforming an arbitrary collection of rational individual preferences into a rational social preference\footnote{Arrow assumes that a rational social preference, like rational individual preferences, must be transitive and complete.} is subject to the following constraint (see Section \ref{ArrovianSection} for a formal statement): if the method respects unanimous individual strict preferences, as required by the \textit{Pareto principle}, and makes the social preference on two alternatives independent of individual preferences on other alternatives, as required by \textit{Independence of Irrelevant Alternatives} (IIA), then the method is a \textit{dictatorship}. Here dictatorship means that there is a single individual whose strict preference is always copied by the social preference, no matter the preferences of other individuals.

There are broadly two responses to Arrow's  Theorem in social choice theory. First, for many working on welfare economics, the response has been to assume the possibility of interpersonal comparisons of utility. When stated in terms of profiles of individual utility functions rather than individual ordinal preference relations (see \citealt[p.~377]{Sen2017}), Arrow's Theorem  assumes that individual utility is only unique up to positive linear transformation, as in expected utility theory.  If one instead assumes that individual utility functions are interpersonally comparable in such a way that they cannot be independently modified by any positive linear transformations, this opens the way to social welfare functionals such as utilitarianism, Rawlsian leximin, and others, which satisfy the Pareto principle and IIA (see \citealt{d'Aspremont2002} and \citealt{Bossert2004}).

Second, for many working on voting theory, the response has been to keep the ordinal representation of individual preference, which matches real election ballots in which voters rank the candidates,\footnote{Even if interpersonal comparisons of utility were possible, it is not clear that preference intensity ought to play a role in voting---see \citealt[pp.~30-31]{Schwartz1986}.} and to advocate voting rules that drastically violate IIA.\footnote{\label{DomainRestrictions}Social choice theorists have also considered voting rules that satisfy all of Arrow's axioms except for \textit{universal domain}; these rules are undefined on some collections of voter ballots, e.g., as in simple majority rule defined only on collections of voter ballots that produce a transitive majority relation (see \citealt[\S~VII.2]{Arrow1963}). For a survey of results on domain restrictions for voting, see \citealt{Gaertner2001}. While these results illuminate the conditions under which challenges for voting such as majority cycles do not arise, the idea to ``use a voting rule with a restricted domain'' does not provide an escape from the practical problem posed by Arrow. We cannot expect every democratic institution to write into its constitution that if voters' ballots  do not together meet some mathematical condition, then all votes are  rejected and no election result delivered.} Examples of such rules include those of Borda \citeyearpar{Borda1781}, Hare \citeyearpar{Hare1859} (also known as Alternative Vote, Instant Runoff Voting, or Ranked Choice Voting), and Copeland \citeyearpar{Copeland1951}, among many others (see \citealt{Brams2002} and  \citealt{Pacuit2019}). In this paper, we are interested in the voting interpretation of Arrow's Theorem. But we argue that escaping Arrow-style impossibility theorems for preferential voting does not require so drastically violating the intuitions behind IIA. We propose a new axiom of \textit{advantage-standard} (AS) \textit{representability} as a replacement for IIA that arguably captures what is right about IIA but revises what is wrong with it. Our argument for AS representability as opposed to IIA is based on a model of comparing candidates in elections that we call the Advantage-Standard model.

The requirement that a collective choice rule (CCR) be representable by the AS model captures a key insight of IIA but is weaker than IIA; yet it is stronger than what is known in the literature as weak IIA (two profiles alike on $x,y$ cannot have opposite strict social preferences on $x$ and~$y$). In addition to motivating violations of IIA, the AS model makes intelligible violations of another Arrovian assumption: the negative transitivity of the strict social preference relation $P$. While  previous literature shows that only weakening IIA to weak IIA or only weakening negative transitivity of $P$ to acyclicity still leads to impossibility theorems (see Section \ref{ArrovianSection}), we show that jointly weakening IIA to AS representability and weakening negative transitivity of $P$ leads to no such impossibility theorems. Indeed, we show that several appealing CCRs are AS representable, including even transitive~CCRs.

The paper is organized as follows. In Section \ref{ArrovianSection}, we recall the framework of Arrovian social choice theory and some impossibility theorems, starting with Arrow's Theorem. In Section~\ref{ASmodelsection}, we introduce  the AS model, our proposed replacement of IIA---AS representability---and our main technical result: the equivalence between AS representability  and the conjunction of weak IIA and a property we call \textit{orderability}. In Section \ref{CCRsection}, we show that three CCRs based on voting methods proposed in the literature are AS representable: a CCR based on the \textit{covering} relation of \citealt{Gillies1959} (cf.~\citealt{Miller1980} and \citealt{Duggan2013}); a CCR variant of the Ranked Pairs voting method introduced in \citealt{Tideman1987}; and the Split Cycle CCR studied in \citealt{HP2020a,HP2020b}. In Section \ref{explanatory}, we show how the AS model can be used to explain properties of AS representable CCRs. We conclude in Section~\ref{Conclusion}.

With a few exceptions, proofs of results are given in the Appendix.

\section{The Arrovian Framework}\label{ArrovianSection}
Let $X$ and $V$ be  nonempty finite sets of \textit{candidates} and \textit{voters}, respectively. Let $R$ be a binary relation on some set $Y\subseteq X$. We write `$xRy$' for $\langle x,y\rangle\in R$, and we define binary relations $P(R)$, $I(R)$, and $N(R)$ on $Y$ as follows:
\begin{itemize}
    \item  $xP(R)y$ if and only if $xRy$ and \textit{not} $yRx$ (strict preference);
    \item $xI(R)y$ if and only if $xRy$ and $yRx$ (indifference);
    \item $xN(R)y$ if and only if neither $xRy$ nor $yRx$ (noncomparability).
\end{itemize}
We say  $R$ is \textit{reflexive} if for all $x\in Y$, $xRx$; $R$ is \textit{complete} if for all $x,y\in Y$, $xRy$ or $yRx$; $R$ is \textit{transitive} if for all $x,y,z\in Y$, if $xRy$ and $yRz$, then $xRz$; and $R$ is \textit{acyclic} if there is no sequence $x_1,\dots,x_m\in Y$ with $m\geq 2$ such that $x_iP(R)x_{i+1}$ for each $i<m$ and $x_m=x_1$. 

Let $B(Y)$ denote the set of all binary relations on $Y$ and $O(Y)$  the set of all transitive and complete binary relations on $Y$. For $Y\subseteq X$, a \textit{$Y$-preprofile}  is a function $\mathbf{R}:V\to B(Y)$; for each $i\in V$, we write `$\mathbf{R}_i$' for $\mathbf{R}(i)$. A \textit{$Y$-profile} is a $Y$-preprofile such that $\mathbf{R}_i\in O(Y)$ for each $i\in V$. When $Y=X$, we simply speak of \textit{preprofiles} and \textit{profiles}. For a profile $\mathbf{R}$  and $Y\subseteq X$,  let $\mathbf{R}|_Y$ be the $Y$-profile  given by $\mathbf{R}|_Y(i)=\mathbf{R}(i)\cap Y^2$ for each $i\in V$.

A \textit{collective choice rule} (CCR) is a function $f:D\to B(X)$, where $D$ is the set of all profiles. Thus, we build the assumption of \textit{universal domain} into the definition of a CCR.\footnote{How our main results would be affected by allowing CCRs with restricted domains (recall Footnote \ref{DomainRestrictions}) is an interesting question that we leave for future work.} We say that a CCR $f$ is \textit{complete} (resp.~\textit{transitive}, \textit{acyclic}) if $f(\mathbf{R})$ is complete (resp.~transitive, acyclic) for all profiles $\mathbf{R}$. Following Arrow \citeyearpar{Arrow1951}, a CCR that is both transitive and complete is called a \textit{social welfare function} (SWF).

We  recall the following standard properties that a CCR might satisfy:
\begin{itemize}
    \item $f$ satisfies \textit{independence of irrelevant alternatives} (IIA) if for all $x,y\in X$ and profiles $\mathbf{R},\mathbf{R}'$, if $\mathbf{R}|_{\{x,y\}}=\mathbf{R}'|_{\{x,y\}}$, then $x f(\mathbf{R})y$ if and only if $x f(\mathbf{R}')y$;
   
    \item  $f$ satisfies \textit{Pareto} if for all $x,y\in X$ and profiles $\mathbf{R}$, if $xP(\mathbf{R}_i)y$ for all $i\in V$, then $x P(f(\mathbf{R}))y$;
    
    \item  $f$ satisfies \textit{strong Pareto} if for all $x,y\in X$ and profiles $\mathbf{R}$, if $x\mathbf{R}_iy$ for all $i\in V$ and $xP(\mathbf{R}_j)y$ for some $j\in V$, then $x P(f(\mathbf{R}))y$; 
    
    \item  $f$ satisfies \textit{Pareto indifference} if for all $x,y\in X$ and profiles $\mathbf{R}$, if $xI(\mathbf{R}_i)y$ for all $i\in V$, then $x I(f(\mathbf{R}))y$; 

    \item $f$ satisfies \textit{anonymity} if for any profile $\mathbf{R}$ and permutation $\tau$ of $V$, \[f(\mathbf{R})  = f(\mathbf{R}^\tau),\] where the profile $\mathbf{R}^\tau$ is defined by $\mathbf{R}^\tau(i)=\mathbf{R}(\tau(i))$.
    
    \item $f$ satisfies \textit{neutrality} if for any profile $\mathbf{R}$ and permutation $\pi$ of $X$, \[\pi f(\mathbf{R})  = f(\mathbf{R}^\pi),\] where for any binary relation $R$ on $X$, we set $\pi R=\{\langle \pi(x),\pi(y)\rangle\mid \langle x,y\rangle\in R\}$, and the profile $ \mathbf{R}^\pi$ is defined by $\mathbf{R}^\pi(i)=\pi\mathbf{R}(i)$.

\end{itemize}

Next we recall Arrow's original impossibility theorem. A \textit{dictator for $f$} is a voter $i\in V$ such that for all $x,y\in X$ and profiles $\mathbf{R}$, if $xP(\mathbf{R}_i)y$, then $xP(f(\mathbf{R}))y$. 

\begin{theorem}[\citealt{Arrow1951}]\label{ArrowFirst} Assume $|X|\geq 3$. If $f$ is an SWF satisfying IIA and Pareto, then there is a dictator for $f$.
\end{theorem}

Arrow's assumption that $f$ is an SWF can be strictly weakened, while keeping all the other axioms the same. A binary relation $P$ on $X$ is \textit{negatively transitive} if for all $x,y,z\in X$, if \textit{not} $xPy$ and \textit{not} $yPz$, then \textit{not} $xPz$. In place of Arrow's assumption that $f(\mathbf{R})$ is transitive and complete, we need only make the assumption that $P(f(\mathbf{R}))$ is negatively transitive, which is strictly weaker. The following is easy to check (for part \ref{part2}, consider a set $X=\{x,y,z\}$ with the reflexive relation $R$ such that $xI(R)y$ and $yI(R)z$ but $xN(R)z$).

\begin{lemma} $\,$
\begin{enumerate}
\item If $R$ is a transitive and complete relation, then $P(R)$ is negatively transitive.
\item\label{part2} There are reflexive relations $R$ such that $P(R)$ is negatively transitive but $R$ is neither complete nor transitive.
\end{enumerate}
\end{lemma} 

Here is the strengthening of Arrow's original impossibility theorem in terms of negative transitivity.\footnote{Some authors, such as Fishburn \citeyearpar{Fishburn1970}, reformulate Arrow's setup so that the output of an SWF $f$ is not a weak relation $R$ but rather a strict relation $P$; then assuming that $P$ is asymmetric and negatively transitive, one can prove the analogous version of Theorem \ref{ArrowFirst}. This is closely related to Proposition \ref{StrongerArrow}, but we remain in Arrow's original setting where the output of $f$ is a weak (and hence not assumed to be asymmetric) relation $R$. This setting is more expressive, since one can distinguish between indifference ($xRy$ and $yRx$) and noncomparability (neither $xRy$ nor $yRx$), a distinction that is lost in Fishburn's setup.}
\begin{proposition}\label{StrongerArrow} Assume  $|X|\geq 3$. If a CCR $f$ satisfies IIA and Pareto, and for every profile $\mathbf{R}$, $P(f(\mathbf{R}))$ is negatively transitive, then there is a dictator for $f$.
\end{proposition}

\begin{proof} Let $f$ satisfy the hypothesis. Define $f'$ such that for all profiles $\mathbf{R}$, $xf'(\mathbf{R})y$ if and only if \textit{not} $yP(f(\mathbf{R}))x$. Then $f'(\mathbf{R})$ is complete, and as $P(f(\mathbf{R}))$ is negatively transitive, $f'(\mathbf{R})$ is transitive. Thus, $f'$ is an SWF. Since $P(f'(\mathbf{R}))=P(f(\mathbf{R}))$, that $f$ satisfies Pareto  implies that $f'$ does as well. Finally, $f'$ satisfies IIA: for $\mathbf{R},\mathbf{R}'$ and $x,y\in X$, if $\mathbf{R}|_{\{x,y\}}=\mathbf{R}'|_{\{x,y\}}$, then by IIA for $f$, we have $x f(\mathbf{R})y$ if and only if $x f(\mathbf{R}')y$, which implies $x f'(\mathbf{R})y$ if and only if $x f'(\mathbf{R}')y$ by the definition of $f'$. Thus, $f'$ is an SWF satisfying IIA and Pareto, so by Arrow's Theorem, there is a dictator $i$ for $f'$. Then $i$ is also a dictator for~$f$.\end{proof}

Weakening the Pareto assumption of Proposition \ref{StrongerArrow} does not significantly improve the situation, by the Murakami-Wilson Theorem. A CCR $f$ satisfies \textit{strict non-imposition} (SNI) if for all $x,y\in X$ with $x\neq y$, there is a profile $\mathbf{R}$ such that $xP(f(\mathbf{R}))y$.\footnote{See \citealt{HK2020} for a proof that this assumption of Murakami \citeyearpar{Murakami1968} is equivalent to the conjunction of the non-null and non-imposition assumptions of Wilson \citeyearpar{Wilson1972} for transitive CCRs.} An \textit{inverse dictator for $f$} is a voter $i\in V$ such that for all $x,y\in X$ and profiles $\mathbf{R}$, if $xP(\mathbf{R}_i)y$, then $yP(f(\mathbf{R}))x$. 

\begin{theorem}[\citealt{Murakami1968}, \citealt{Wilson1972}] Assume  $|X|\geq 3$. If $f$ is an SWF satisfying IIA and SNI, then there is a dictator for $f$ or an inverse dictator for $f$.
\end{theorem}

By the same reasoning as in the proof of Proposition \ref{StrongerArrow}, we obtain the following.

\begin{proposition} Assume $|X|\geq 3$. If $f$ is a CCR satisfying IIA and SNI such that for every profile $\mathbf{R}$, $P(f(\mathbf{R}))$ is negatively transitive, then there is a dictator for $f$ or an inverse dicatator for $f$.
\end{proposition}
\noindent Even dropping SNI is not much help, as there is still a partition of the candidates with a dictator or inverse dictator over each cell of the partition and relations between the cells imposed independently of voter preferences (see \citealt[\S~3]{Wilson1972}).

Thus, we must weaken negative transitivity or IIA or both. Just weakening negative transitivity to acyclicity still does not deliver us out of the landscape of impossibility theorems due to IIA, as shown by vetoer theorems in \citealt{Blau1977}, \citealt{Blair1979}, and \citealt{Kelsey1984,Kelsey1984b,Kelsey1985}.\footnote{For example, Blau and Deb \citeyearpar{Blau1977} show that under IIA, acyclicity, neutrality, and monotonicity, any partition of $V$ into at most  $|X|$-many coalitions contains a coalition with veto power, meaning that if every voter in the coalition ranks $x$ above $y$, then society cannot rank $y$ above $x$. Hence when $|X|\geq |V|$, there is a single voter with veto power. The later cited papers  replace neutrality and monotonicity with other weak assumptions and derive related vetoer results.\label{acylicity}} Although variants of the unanimity CCR, with $xf(\mathbf{R})y$ if and only if all voters have $x\mathbf{R}_iy$ (see \citealt{Weymark1984}), satisfy transitivity and IIA and may be reasonable for some small committees in which unanimity is valued and possible, it is not practical for every voter to have veto power in settings in which some conflicts of preferences are bound to arise.\footnote{Similar points apply to the acyclic supermajority CCR with $xf(\mathbf{R})y$ if and only if  $|\{i\in V\mid yP(\mathbf{R}_i)x\}| \leq\frac{|X|-1}{|X|}\times |\{i\in V\mid xP(\mathbf{R}_i)y\mbox{ or }yP(\mathbf{R}_i)x\}|$. For more than a few alternatives, the supermajority threshold  $\frac{|X|-1}{|X|}$ is impractically high for many voting contexts, and any lowering of the threshold that allows for more strict social preferences leads to violations of acyclicity (see \citealt[Thm.~4]{Ferejohn1974}).} In the end, to make room for  voting methods  that are practical for a wide range of numbers of voters and candidates, we must weaken IIA. Ideally, we would do so while retaining some of the insight of IIA that made it appealing in the first place. We pursue this approach in the next section. 

\section{The Advantage-Standard Model}\label{ASmodelsection}
In this section, we introduce the Advantage-Standard model. Roughly, a CCR can be represented within this model if the strict social preference according to the CCR can be viewed as arising from a comparison of implicit \textit{advantage} and \textit{standard} functions. In a  profile $\mathbf{R}$, we  assess the advantage that one candidate $x$ has over another candidate $y$ just in terms of how voters rank $x$ vs.~$y$, i.e., just in terms of $\mathbf{R}|_{\{x,y\}}$. No other candidates are relevant to assessing the advantage of $x$ over $y$. This, we think, is part of the intuition and insight behind IIA. However, whether the intrinsic advantage that $x$ has over $y$ in the profile is sufficient to judge that $xP(f(\mathbf{R}))y$  depends on a standard for strict preference that is contextually determined, where the context is the rest of the profile $\mathbf{R}$ besides $\mathbf{R}|_{\{x,y\}}$ (the preprofile $\mathbf{R}^{-x,y}$ in Definition \ref{advstddef} below). The mistake of IIA is to assume that the standard for strict preference is the same for all contexts. For example, in the left profile in Figure \ref{IIAfig}, it is perfectly reasonable to hold that the advantage of $a$ over $b$ exceeds the standard for a strict social preference of $a$ over $b$. Yet in the right profile in Figure \ref{IIAfig}, where the advantage of $a$ over $b$ is the same, the standard cannot be met: society must have full indifference or noncomparability between $a$, $b$, and $c$, assuming anonymity, neutrality, and acyclicity, given the symmetries in the context. Thus, the standard must be context dependent. 

Context dependent standards are familiar in other domains. For example, whether a person counts as ``tall'' depends on who else is being assessed for tallness within the context of judgment; whether a performance on an exam counts as ``passing'' may depend on other performances in the class; and so on. In the setting of this paper, whether a candidate's advantage over another suffices for strict social preference depends on which other advantages are being assessed for strict social preference within the context of an~election. 

\begin{figure}
\begin{center}
\begin{tabular}{ccc}
$i$ & $j$ & $k$   \\\hline
$\boldsymbol{a}$ & $\boldsymbol{b}$ &  $\boldsymbol{a}$ \\
$\boldsymbol{b}$ &  $\boldsymbol{a}$ & $\boldsymbol{b}$ \\
$c$ &  $c$ &  $c$ \\
\end{tabular}\qquad\qquad \begin{tabular}{ccc}
$i$ & $j$ & $k$   \\\hline
$\boldsymbol{a}$ & $\boldsymbol{b}$ &  $c$ \\
$\boldsymbol{b}$ &  $c$ & $\boldsymbol{a}$ \\
$c$ &  $\boldsymbol{a}$ &  $\boldsymbol{b}$ \\
\end{tabular}
\end{center}\caption{two profiles $\mathbf{R}$ and $\mathbf{R'}$ such that $\mathbf{R}|_{\{a,b\}}=\mathbf{R}'|_{\{a,b\}}$. The first column indicates that voter $i$ has $aP(\mathbf{R}_i)b$, $bP(\mathbf{R}_i)c$, etc. }\label{IIAfig}
\end{figure}

We will measure the advantage of one candidate over another and the standard required for strict preference using a totally ordered set of degrees, e.g., the integers or rational numbers with their usual ordering. In order to subsume different possible definitions of advantage and standard, we allow advantages and standards to be elements of any \textit{totally ordered group} $(G,\circ, \leq)$ (see  Appendix \ref{GroupAppendix} for a definition). Natural choices would be $(\mathbb{Z},+,\leq)$, the additive group of the integers with the usual order, or $(\mathbb{Q}_{>0},\times,\leq)$, the multiplicative group of positive rational numbers with the usual order. A key principle is that the advantage of $x$ over $y$ is the \textit{inverse} in the group of the advantage of $y$ over $x$, with the inverse denoted by $(\cdot)^{-1}$. E.g., working with $(\mathbb{Z},+,\leq)$, the advantage of $x$ over $y$ is the negative of the advantage of $y$ over $x$;\footnote{Compare this with the definition of a \textit{comparison function} in \citealt{Dutta1999}.}  working with $(\mathbb{Q}_{>0},\times,\leq)$, the advantage of $x$ over $y$ is the reciprocal of the advantage of $y$ over $x$. A second key principle is that the standard for $x$ to defeat $y$ must be at least as large as the identity element of the group, denoted by $e$. E.g., working with $(\mathbb{Z},+,\leq)$, the standard must be at least 0; working with $(\mathbb{Q}_{>0},\times,\leq)$, the standard must be at least 1.  This ensures that at most one of the advantage of $x$ over $y$ and the advantage of $y$ over $x$ exceeds the standard for strict preference.

With this setup, we are ready for the definition of AS representability.

\begin{definition}\label{advstddef}
For $x,y\in X$ and a profile $\mathbf{R}$, define $\mathbf{R}^{-x,y}:V\to B(X)$ by $\mathbf{R}^{-x,y}(i)=\mathbf{R}(i)\setminus \{\langle x,y\rangle,\langle y,x\rangle\}$. Let
\[D_A:=\{\langle x,y,\mathbf{R}|_{\{x,y\}}\rangle \mid x,y\in X,\mathbf{R}\mbox{ a profile}\}\]
and
\[D_S:=\{\langle x,y,\mathbf{R}^{-x,y}\rangle \mid x,y\in X, \mathbf{R}\mbox{ a profile}\}.\]
A CCR $f$ is \textit{advantage-standard} (AS) \textit{representable} if there exists a totally ordered group $(G,\circ, \leq)$ with identity element $e$ and functions $\mathsf{Advantage}:D_A\to G$ and $\mathsf{Standard}:D_S\to G$ such that for all $x,y\in X$ and profiles $\mathbf{R}$:
\begin{equation}\mathsf{Advantage}(x,y,\mathbf{R}|_{\{x,y\}})=\mathsf{Advantage}(y,x, \mathbf{R}|_{\{x,y\}})^{-1};\label{minus}\end{equation}
\begin{equation}\mathsf{Standard}(x,y,\mathbf{R}^{-x,y})\geq e;\label{minimal}\end{equation}
and, where $x>y$ if and only if $x\geq y$ and $x\neq y$,
\begin{equation} xP(f(\mathbf{R}))y\iff \mathsf{Advantage}(x,y,\mathbf{R}|_{\{x,y\}})> \mathsf{Standard}(x,y,\mathbf{R}^{-x,y}).\label{Piff}\end{equation}
\end{definition}
\noindent Note that we will refer back to (\ref{minus}), (\ref{minimal}), and (\ref{Piff}) without citing Definition \ref{advstddef}.\footnote{(\ref{minus}) and (\ref{minimal}) are the minimal conditions required for AS representability. One could consider additional conditions on the advantage and standard functions. A natural condition on \textsf{Advantage} is: if $\mathbf{R}'$ is obtained from $\mathbf{R}$ by only moving $x$ up in some voter's ballot, then the advantage of $x$ over any $y$ should not decrease from $\mathbf{R}$ to $\mathbf{R}'$. On \textsf{Standard}, one might impose a continuity condition requiring that small differences between $\mathbf{R}$ and $\mathbf{R}'$ be matched by small differences between $f(\mathbf{R})$ and $f(\mathbf{R}')$ (see Definition \ref{ContinuousDef} below).}

One natural measure of the advantage of $x$ over $y$ is the margin of victory:
\[\mathrm{Margin}_{\mathbf{R}}(x,y)=|\{i\in V\mid xP(\mathbf{R}_i)y \}| - |\{i\in V\mid yP(\mathbf{R}_i)x \}|.\]
Another natural measure, which can produce a different ordering when voters' ballots contain ties, is the following:
\[\mathrm{Ratio}_{\mathbf{R}}(x,y)=\begin{cases} \frac{|\{i\in V\mid xP(\mathbf{R}_i)y \}| }{ |\{i\in V\mid yP(\mathbf{R}_i)x \}|} &\mbox{if } |\{i\in V\mid xP(\mathbf{R}_i)y \}|\neq 0\mbox{ and }|\{i\in V\mid yP(\mathbf{R}_i)x \}|\neq 0  \\ 
|V| &\mbox{if }|\{i\in V\mid xP(\mathbf{R}_i)y \}|\neq 0\mbox{ and }|\{i\in V\mid yP(\mathbf{R}_i)x \}|=0\\
\frac{1}{|V|} &\mbox{if }|\{i\in V\mid xP(\mathbf{R}_i)y \}|= 0\mbox{ and }|\{i\in V\mid yP(\mathbf{R}_i)x \}|\neq 0\\
1 &\mbox{if }|\{i\in V\mid xP(\mathbf{R}_i)y \}|= 0\mbox{ and }|\{i\in V\mid yP(\mathbf{R}_i)x \}|=0.\end{cases}\]
The second and third cases are related to the strong Pareto principle; they formalize the idea that the advantage of $x$ over $y$ is maximal when all voters weakly prefer $x$ to $y$ and some voter strictly prefers $x$ to $y$. Of course, in such cases one may also wish to take into account how many voters strictly prefer $x$ to $y$, replacing $|V|$ (resp.~$1/|V|$) in the definition by $|V|+|\{i\in V\mid xP(\mathbf{R}_i)y \}|$ (resp.~$1/(|V|+|\{i\in V\mid yP(\mathbf{R}_i)x \}|)$). The exact formulation of $\mathrm{Ratio}_\mathbf{R}$ will not matter for us here. 

To see the difference between the margin and ratio approaches, suppose that (a) $70$ voters have $xP(\mathbf{R}_i)y$ and $30$ have $yP(\mathbf{R}_i)x$ or (b) $8$ voters have $xP(\mathbf{R}_i)y$, $2$ have $yP(\mathbf{R}_i)x$, and $90$ have  $xI(\mathbf{R}_i)y$. According to $\mathrm{Margin}_\mathbf{R}$, the advantage of $x$ over $y$ is greater in (a), whereas according to  $\mathrm{Ratio}_\mathbf{R}$, the advantage of $x$ over $y$ is greater in (b). For an  example illustrating the second case in the definition of $\mathrm{Ratio}_\mathbf{R}$, suppose that (i) $99$ voters have $xP(\mathbf{R}_i)y$ and $1$ has $yP(\mathbf{R}_i)x$ or (ii) $1$ voter has $xP(\mathbf{R}_i)y$ and $99$ have $xI(\mathbf{R}_i)y$. According to $\mathrm{Margin}_\mathbf{R}$, the advantage of $x$ over $y$ is greater in  (i), whereas according to $\mathrm{Ratio}_\mathbf{R}$, the advantage of $x$ over $y$ is greater in  (ii).

\begin{example}
Consider again the totally ordered group ($\mathbb{Z}, +,\leq)$. Define 
\[\mathsf{Advantage}(x,y,\mathbf{R}|_{\{x,y\}}):=\mathrm{Margin}_{\mathbf{R}|_{\{x,y\}}}(x,y)\]
and
\[\mathsf{Standard}(x,y,\mathbf{R}^{-x,y}):=0.\]
In this case, the intrinsic advantage of $x$ over $y$ is the majority margin, and the standard is constant across contexts. Note that $\mathsf{Advantage}$ and $\mathsf{Standard}$ satisfy (\ref{minus}) and (\ref{minimal}). Then a CCR which outputs the majority relation---so $xP(f(\mathbf{R}))y$ if and only if strictly more voters rank $x$ above $y$ than $y$ above $x$---is AS representable relative to $\mathsf{Advantage}$ and $\mathsf{Standard}$. Of course, this CCR is not acyclic when $|X|\geq 3$, but we will give acyclic examples in Section~\ref{CCRsection}.
\end{example}

The first point to make about AS representability is that it is a weakening of IIA.

\begin{proposition}\label{IIAtoAS} If a CCR $f$ satisfies IIA, then it is AS representable.
\end{proposition}

\begin{proof} Suppose $f$ satisfies IIA. For any $(x,y,\mathbf{Q})\in D_A$ and $(x,y,\mathbf{T})\in D_S$, we define
\begin{eqnarray*}
\mathsf{Advantage}(x,y,\mathbf{Q}) & = & \begin{cases} 1 &\mbox{if }xP(f(\mathbf{S}))y\mbox{ for some profile }\mathbf{S}\mbox{ with }\mathbf{S}|_{\{x,y\}}=\mathbf{Q} \\
-1 &\mbox{if }yP(f(\mathbf{S}))x\mbox{ for some profile }\mathbf{S}\mbox{ with }\mathbf{S}|_{\{x,y\}}=\mathbf{Q} \\
0 &\mbox{otherwise}
\end{cases} \\
\mathsf{Standard}(x,y,\mathbf{T}) &=& 0.
\end{eqnarray*}
That (\ref{minus}) and (\ref{minimal}) hold is immediate. To verify (\ref{Piff}), consider any $x,y\in X$ and profile $\mathbf{R}$. If $xP(f(\mathbf{R}))y$, then by the definitions above, taking $\mathbf{S}=\mathbf{R}$,  $\mathsf{Advantage}(x,y,\mathbf{R}|_{\{x,y\}})=1$ and hence $\mathsf{Advantage}(x,y,\mathbf{R}|_{\{x,y\}})> \mathsf{Standard}(x,y,\mathbf{R}^{-x,y})$.  If $\mathsf{Advantage}(x,y,\mathbf{R}|_{\{x,y\}})> \mathsf{Standard}(x,y,\mathbf{R}^{-x,y})$, then $\mathsf{Advantage}(x,y,\mathbf{R}|_{\{x,y\}})=1$, so there is some profile $\mathbf{S}$ such that $xP(f(\mathbf{S}))y$ and $\mathbf{S}|_{\{x,y\}}=\mathbf{R}|_{\{x,y\}}$. Then by IIA, we have $xP(f(\mathbf{R}))y$.\end{proof}

We will see in Section \ref{CCRsection} that AS representability does not imply IIA. However, it does imply the following weakening of IIA due to Baigent \citeyearpar{Baigent1987}.

\begin{definition}\label{WeakIIA} A CCR $f$ satisfies \textit{weak IIA} if for all $x,y\in X$ and profiles $\mathbf{R},\mathbf{R}'$ such that $\mathbf{R}|_{\{x,y\}}=\mathbf{R}'|_{\{x,y\}}$:
\[xP(f(\mathbf{R}))y\implies \mbox{ \textit{not} } yP(f(\mathbf{R}'))x.\]
\end{definition}

\begin{remark} Another definition of weak IIA, which is equivalent to Definition \ref{WeakIIA} for complete relations (as Baigent assumed $f(\mathbf{R})$ to be) but not in general, has $xP(f(\mathbf{R}))y\implies  xf(\mathbf{R}')y$. However, if we do not initially assume completeness of social preference, then we must be careful to separate constraints on the strict social preference relation from those on the weak social preference relation. We will define two strengthenings of weak IIA that impose constraints on the weak social preference relation in Section \ref{CoveringSection}.
\end{remark}

\begin{proposition}\label{PIweakIIAlemma}
If $f$ is AS representable, then $f$ satisfies weak IIA.
\end{proposition}
\begin{proof}
Assume $f$ is AS representable with functions $\mathsf{Advantage}$ and $\mathsf{Standard}$, which take values in a totally ordered group $(G,\circ,\leq)$ with identity element $e$. Let $x,y\in X$ and $\mathbf{R},\mathbf{R}'$ be two profiles with $\mathbf{R}|_{\{x,y\}}=\mathbf{R}'|_{\{x,y\}}$. Assume $xP(f(\mathbf{R}))y$. By (\ref{Piff}), this implies 
\[ \mathsf{Advantage}(x,y,\mathbf{R}|_{\{x,y\}})>\mathsf{Standard}(x,y,\mathbf{R}^{-x,y}).\]
With (\ref{minimal}), this implies that $\mathsf{Advantage}(x,y,\mathbf{R}|_{\{x,y\}})>e$ and so \[\mathsf{Advantage}(x,y,\mathbf{R}'|_{\{x,y\}})=\mathsf{Advantage}(x,y,\mathbf{R}|_{\{x,y\}})>e.\]
Since $G$ is a totally ordered group, for any $a\in G$, we have that $e< a$ implies  $a^{-1}<e$ (see Lemma \ref{GroupLem}). Thus, (\ref{minus}), (\ref{minimal}), and the previous line together imply that
\[\mathsf{Advantage}(y,x,\mathbf{R}'|_{\{x,y\}})<e\leq \mathsf{Standard}(y,x,\mathbf{R}'^{-x,y}).\]
It follows  by (\ref{Piff}) that \textit{not} $yP(f(\mathbf{R}'))x$.\end{proof}

However, as we will see, the converse of Proposition \ref{PIweakIIAlemma} does not hold. The missing piece is that  AS representability requires that restricted profiles on a pair $x,y$ can be \textit{ordered} (allowing for ties) by strength of the intrinsic advantage of $x$ over $y$ \textit{in a context-independent way}.  Weak IIA, on the other hand, is consistent with there being no such context-independent ordering: it is consistent with weak IIA that there be two restricted profiles $\mathbf{R}|_{\{x,y\}}$ and $\mathbf{R}'|_{\{x,y\}}$ such that in one context, $\mathbf{R}|_{\{x,y\}}$ suffices for a strict social preference for $x$ over $y$ while $\mathbf{R}'|_{\{x,y\}}$ does not, and in another context, $\mathbf{R}'|_{\{x,y\}}$ suffices for a strict social preference for $x$ over $y$ while $\mathbf{R}|_{\{x,y\}}$ does not. By adding to weak IIA the following condition of \textit{orderability}, we rule out this kind of non-uniformity and ensure that restricted profiles can be ordered by strength of intrinsic advantage in a context-independent way, which will lead to a characterization of AS representability.

\begin{definition}\label{orderable}
Let $f$ be a CCR. Let $\mathcal{P}^+(x,y)$ be the set of all $\{x,y\}$-profiles $\mathbf{Q}$ for which there is some profile $\mathbf{R}$ with $\mathbf{R}|_{\{x,y\}}=\mathbf{Q}$ and $xP(f(\mathbf{R}))y$. We say that $f$ is \textit{orderable} if the following holds for every $x,y\in X$:
\begin{itemize}
   \item There is a transitive and complete relation $\leqslant$ on $\mathcal{P}^+(x,y)$ such that for any profiles $\mathbf{R}^1$ and $\mathbf{R}^2$, if $\mathbf{R}^1|_{\{x,y\}}\leqslant \mathbf{R}^2|_{\{x,y\}}$ and $\mathbf{R}^{1-x,y}=\mathbf{R}^{2-x,y}$, then  $xP(f(\mathbf{R}^1))y$ implies $xP(f(\mathbf{R}^2))y$.
\end{itemize}
\end{definition}

Orderability is perhaps best understood by seeing how it can fail, as in the following example. Readers eager for the punchline on orderability may skip to Proposition \ref{converse}.

\begin{example}\label{DodgsonExample} We consider a CCR based on the well-known Dodgson voting rule (\citealt{Fishburn1977,Fishburn1982}, \citealt{Brandt2009}, \citealt{Caragiannis2016}). For simplicity, we only define the Dodgson CCR for profiles in which each voter submits a linear order of the candidates, i.e., with no indifference. Since it is not orderable with respect to this limited domain, no extension of the CCR to arbitrary profiles is  orderable either. A linear profile $\mathbf{R}'$ is obtained from another $\mathbf{R}$ by an \textit{adjacent inversion} if one voter flips two candidates immediately next to each other in her ranking.\footnote{Formally, there is some $i\in V$ and $x,y\in X$ such that $x P(\mathbf{R}_i) y$, and for all $z\in X\setminus\{x,y\}$, either $z P(\mathbf{R}_i) x$ or $yP(\mathbf{R}_i) z$; and $\mathbf{R}'$ differs from $\mathbf{R}$ only in that $yP(\mathbf{R}_i')x$.} The \textit{Dodgson score} of a candidate $x$ in $\mathbf{R}$ is the minimal number of adjacent inversions needed to obtain a profile $\mathbf{R}'$ in which $x$ is the \textit{Condorcet winner}, i.e., in which $\mathrm{Margin}_\mathbf{R}(x,y)>0$ for all $y\in X\setminus \{x\}$. The Dodgson CCR is defined by \[\mbox{$xf(\mathbf{R})y$ if the Dodgson score of $x$ is less than or equal to that of $y$.}\]

 The Dodgson CCR is not orderable. To see this, we consider a profile $\mathbf{R}$  identified by Fishburn \citeyearpar[p.~132]{Fishburn1982}, where the number above a ranking indicates how many voters have that ranking:
\begin{center}
$\mathbf{R}$\quad \begin{tabular}{ccccccc}
$3$ & $2$ & $9$ & $5$ & $9$ & $13$ &$2$   \\\hline
$y$ & $y$ &$y$ &  $x$ & $x$& $z$ & $\boldsymbol{z}$ \\
$x$ & $\boldsymbol{x}$ & $w$ & $z$ & $y$& $x$ & $\boldsymbol{x}$ \\
$z$ & $\boldsymbol{z}$ & $z$ &  $y$ & $w$& $w$ & $w$ \\
$w$ & $w$ & $x$ &  $w$ & $z$& $y$ & $y$ \\
\end{tabular}\qquad\qquad $\mathbf{R}'$\quad\begin{tabular}{ccccccc}
$3$ & $2$ & $9$ & $5$ & $9$ & $13$ &$2$   \\\hline
$y$ & $y$ &$y$ &  $x$ & $x$& $z$ & $\boldsymbol{x}$ \\
$x$ & $\boldsymbol{z}$ & $w$ & $z$ & $y$& $x$ & $\boldsymbol{z}$ \\
$z$ & $\boldsymbol{x}$ & $z$ &  $y$ & $w$& $w$ & $w$ \\
$w$ & $w$ & $x$ &  $w$ & $z$& $y$ & $y$ \\
\end{tabular} \vspace{.1in}

$\mathbf{S}$\quad \begin{tabular}{ccccccc}
$3$ & $2$ & $9$ & $5$ & $9$ & $13$ &$2$   \\\hline
$y$ & $\boldsymbol{x}$ &$y$ &  $x$ & $x$& $z$ & $y$ \\
$x$ & $\boldsymbol{z}$ & $w$ & $z$ & $y$& $x$ & $\boldsymbol{z}$ \\
$z$ & $w$ & $z$ &  $y$ & $w$& $w$ & $\boldsymbol{x}$ \\
$w$ & $y$ & $x$ &  $w$ & $z$& $y$ & $w$ \\
\end{tabular}\qquad\qquad $\mathbf{S}'$\quad\begin{tabular}{ccccccc}
$3$ & $2$ & $9$ & $5$ & $9$ & $13$ &$2$   \\\hline
$y$ & $\boldsymbol{z}$ &$y$ &  $x$ & $x$& $z$ & $y$ \\
$x$ & $\boldsymbol{x}$ & $w$ & $z$ & $y$& $x$ & $\boldsymbol{x}$ \\
$z$ & $w$ & $z$ &  $y$ & $w$& $w$ & $\boldsymbol{z}$ \\
$w$ & $y$ & $x$ &  $w$ & $z$& $y$ & $w$ \\
\end{tabular}
\end{center}
 
 In  profile $\mathbf{R}$, the majority margins are:
\begin{itemize}
    \item  $\mathrm{Margin}_\mathbf{R}(x,y)=29-14=15$ and $\mathrm{Margin}_\mathbf{R}(x,w)=34-9=25$;
    \item $\mathrm{Margin}_\mathbf{R}(y,z)=23-20=3$ and $\mathrm{Margin}_\mathbf{R}(y,w)=28-15=13$;
    \item $\mathrm{Margin}_\mathbf{R}(z,x)=24-19=5$ and $\mathrm{Margin}_\mathbf{R}(z,w)=25-18=7$.
\end{itemize}
Thus, in $\mathbf{R}$ there is a cycle ($x$ beats $y$ beats $z$ beats $x$) and hence no Condorcet winner. Now from  $\mathbf{R}$, 3 $zx$ to $xz$ inversions (say, by voters in the third column) will make $x$ majority preferred to $z$ and hence the Condorcet winner. By contrast, at least $4$ adjacent inversions are required to make $z$ majority preferred to $y$ and hence the Condorcet winner, because every ranking that has $y$ above $z$ has a third alternative between them. Thus, $x$'s Dodgson score of 3 is lower than $z$'s score of 4, so $xP(f(\mathbf{R}))z$. 

Now the profile $\mathbf{R}'$ above differs from $\mathbf{R}$ only on $x,z$ (see the spots in bold), i.e., $\mathbf{R}^{-x,z}=\mathbf{R}'^{-x,z}$. Consider what happens when we change the intrinsic advantage between $x$ and $z$ from $\mathbf{R}|_{\{x,z\}}$ to $\mathbf{R}'|_{\{x,z\}}$. In $\mathbf{R}'$, just 2 $yz$ to $zy$ inversions in the second column will make $z$ the Condorcet winner, whereas at least $3$ adjacent inversions are required to make $x$ the Condorcet winner. Thus, $zP(f(\mathbf{R}'))x$. The moral is that \textit{moving from $\mathbf{R}|_{\{x,z\}}$ to $\mathbf{R}'|_{\{x,z\}}$ has helped $z$ against~$x$ in the context of $\mathbf{R}^{-x,z}=\mathbf{R}'^{-x,z}$}, since $xP(f(\mathbf{R}))z$ but $zP(f(\mathbf{R}'))x$.

Yet moving from $\mathbf{R}|_{\{x,z\}}$ to $\mathbf{R}'|_{\{x,z\}}$ can \textit{hurt} $z$ against~$x$ in another context. In the profile $\mathbf{S}$ above, the majority margins are exactly as in $\mathbf{R}$. In $\mathbf{S}$, 2 $yz$ to $zy$ inversions in the last column make $z$ the Condorcet winner, whereas at least 3 adjacent inversions are required to make $x$ the Condorcet winner.  Thus, $zP(f(\mathbf{S}))x$. Now $\mathbf{S}'$ differs from $\mathbf{S}$ only on $x,z$, i.e., $\mathbf{S}^{-x,z}=\mathbf{S}'^{-x,z}$. Moreover, the change in intrinsic advantage from $\mathbf{S}|_{\{x,z\}}$ to $\mathbf{S}'|_{\{x,z\}}$ is exactly the same as above, since $\mathbf{S}|_{\{x,z\}}=\mathbf{R}|_{\{x,z\}}$ and $\mathbf{S}'|_{\{x,z\}}=\mathbf{R}'|_{\{x,z\}}$. Yet now \textit{the same change in intrinsic advantage between $x$ and $z$ hurts $z$ in the context of $\mathbf{S}^{-x,z}=\mathbf{S}'^{-x,z}$}. In $\mathbf{S}'$, while 3 $zx$ to $xz$ inversions will make $x$ the Condorcet winner, at least 4 adjacent inversions are required to make $z$ the Condorcet winner, because every ranking that has $y$ above $z$ has a third alternative between them. Thus, $xP(f(\mathbf{S}'))z$.

Because the same change in intrinsic advantage can help $z$ against $x$ in one context and hurt $z$ against $x$ in another context, the Dodgson CCR is not orderable.\end{example}

A modification of the Dodgson CCR in Example \ref{DodgsonExample} gives a natural example of a CCR satisfying weak IIA but not orderability.

\begin{example} Define the Majority Dodgson CCR $g$ by: $xg(\mathbf{R})y$ if either (i) the Dodgson scores of $x$ and $y$ are equal or (ii)  the Dodgson score of $x$ is less than or equal to that of $y$ and $\mathrm{Margin}_\mathbf{R}(x,y)> 0$. Thus, $xP(g(\mathbf{R}))y$ only if (ii) holds. Since $xP(g(\mathbf{R}))y$ implies $\mathrm{Margin}_\mathbf{R}(x,y)> 0$, it is immediate that $g$ satisfies weak IIA. However, we can see that $g$ is not orderable in the same way we saw that Dodgson is not orderable in Example \ref{DodgsonExample}:   moving from $\mathbf{R}|_{\{x,z\}}$ to $\mathbf{R}'|_{\{x,z\}}$ helps $z$ against~$x$ in the context of $\mathbf{R}^{-x,z}=\mathbf{R}'^{-x,z}$, since $xN(g(\mathbf{R}))z$ (in contrast to $xP(f(\mathbf{R}))z$ for the Dodgson CCR $f$) but $zP(g(\mathbf{R}'))x$; yet the same change in intrinsic advantage hurts $z$ in the context of $\mathbf{S}^{-x,z}=\mathbf{S}'^{-x,z}$, since $zP(g(\mathbf{S}))x$ and  $xN(g(\mathbf{S}'))z$ (in contrast to $xP(f(\mathbf{S}'))z$ for the Dodgson CCR $f$).\end{example}

The example showing that Dodgson is not orderable required 5 candidates and many voters. However, we can show under weaker cardinality assumptions that there are CCRs---even transitive CCRs---satisfing weak IIA but not orderability.

\begin{restatable}{proposition}{converse}\label{converse} Assume $|V|\geq 5$ and $|X|\geq 3$. Then there are transitive CCRs satisfying anonymity, neutrality, Pareto, and weak IIA that are not orderable.
\end{restatable}

Conversely, orderability does not imply weak IIA, as shown by  the following.

\begin{example}\label{CopelandEx} We consider the Copeland CCR (see \citealt{Rubinstein1980}) based on the Copeland voting method (\citealt{Copeland1951}). For a profile $\mathbf{R}$, the \textit{Copeland score} of a candidate $x$ in $\mathbf{R}$ is the number of candidates to whom $x$ is majority preferred minus the number who are majority preferred to $x$:
\[|\{y\in X\mid \mathrm{Margin}_\mathbf{R}(x,y)>0\}| - |\{y\in X\mid \mathrm{Margin}_\mathbf{R}(y,x)>0\}|.\]
The Copeland CCR $f$ is defined by  \[\mbox{$xf(\mathbf{R})y$ if the Copeland score of $x$ is at least that of $y$.}\] It is easy to construct pairs of profiles (assuming $|V|\geq 2$ and $|X|\geq 3$) such that:  all voters rank $x$ and $y$ in the same way; in one profile, $x$ has a higher Copeland score than $y$; and in the other profile, $y$ has a higher Copeland score than $x$. Thus, the Copeland CCR violates weak IIA. Finally, order $\mathcal{P}^+(x,y)$ by $\mathbf{Q}\leqslant \mathbf{Q}'$ if $\mathrm{Margin}_\mathbf{Q}(x,y)\leq \mathrm{Margin}_{\mathbf{Q}'}(x,y)$. The Copeland CCR is clearly orderable with this ordering.\end{example}

An easy further analysis of the properties of the Copeland CCR in Example \ref{CopelandEx} yields the following.

\begin{proposition}\label{OrderableNotWeakIIA} Assume $|V|\geq 2$ and $|X|\geq 3$. Then there are SWFs satisfying anonymity, neutrality, Pareto, and orderability that do not satisfy weak IIA.
\end{proposition}

Having shown that neither weak IIA nor orderability implies the other, we now give our characterization of AS representability in terms of both, the main technical result of this paper.

\begin{restatable}{theorem}{characterization}\label{characterization}
A CCR $f$ is AS representable if and only if $f$  satisfies weak~IIA and orderability.\end{restatable}

Together with Propositions \ref{converse} and \ref{OrderableNotWeakIIA}, Theorem \ref{characterization} implies that neither weak IIA nor orderability alone implies AS representability. The logical relations between IIA, AS representability, weak IIA, and orderability are summarized in Figure \ref{LogicalFig}.\footnote{We leave it as an open problem to find an axiom that can be added to weak IIA and orderability (and is independent of these axioms) to obtain a characterization of IIA (or at least a version of IIA for the strict social relation $P(f(\mathbf{R}))$, which is the relation at issue in weak IIA and orderability).}

\newcommand\xdownarrow[1][2ex]{%
   \mathrel{\rotatebox{90}{$\Leftarrow{\rule{#1}{0pt}}$}}
}

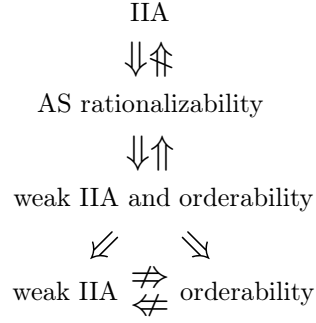
\begin{figure}
\begin{center}
\begin{tikzpicture}[scale=1.25]

\node at (0,0) (IIA) {IIA}; 
\node at (-.15,-.5)  {\scalebox{1.5}{$\Downarrow$}};
\node at (.15,-.5)  {\scalebox{1.5}{\rotatebox{90}{$\nRightarrow$}}};

\node at (0,-1)  {AS representability}; 
\node at (0,-1.5)  {\scalebox{1.5}{$\Downarrow\Uparrow$}};
\node at (0,-2)  {\,\,\,\, weak IIA and orderability}; 
\node at (-.5,-2.5)  {\scalebox{1.5}{\rotatebox{-45}{$\Downarrow$}}}; 

\node at (.5,-2.5) 
{\scalebox{1.5}{\rotatebox{45}{$\Downarrow$}}}; 
\node at (0,-3)  {\,\,\,\, weak IIA \qquad orderability}; 
\node at (0,-2.85)  {\scalebox{1.5}{$\nRightarrow$}}; 
\node at (0,-3.15)  {\scalebox{1.5}{$\nLeftarrow$}};

\end{tikzpicture}
\end{center}
    \caption{logical relations between axioms}\label{LogicalFig}
\end{figure}

The proof that orderability and weak IIA together imply AS representability in fact establishes a stronger conclusion: $f$ is AS representable with integer-valued advantage and standard functions.\footnote{The proof of Theorem \ref{characterization} also provides a canonical construction of the integer-valued advantage and standard functions for a CCR satisfying orderability and weak IIA.} We thus have the following corollary to Theorem~\ref{characterization}.
\begin{corollary}
If a CCR $f$ is AS representable, then $f$ is AS representable with advantage and standard functions taking values in $\mathbb{Z}$ (viewed as a totally ordered group with the usual addition operation and order).
\end{corollary}

\section{Advantage-Standard Representable CCRs}\label{CCRsection}

Does weakening IIA to AS representability make room for appealing CCRs that escape Arrow's Impossibility Theorem and similar results? We think that it does. In this section, we give three~examples. We chose these examples, among other appealing AS representable CCRs, because they have generated or are generating much interest in voting theory.

We note immediately that we should not expect any of the following CCRs to be negatively transitive. First, on the AS model, there is no guarantee that a CCR will be negatively transitive, for there is no \textit{a priori} reason to assume that together
\begin{eqnarray*}
\mathsf{Advantage}(x,y,\mathbf{R}|_{\{x,y\}})\not> \mathsf{Standard}(x,y,\mathbf{R}^{-x,y}) \\
\mathsf{Advantage}(y,z,\mathbf{R}|_{\{y,z\}})\not> \mathsf{Standard}(y,z,\mathbf{R}^{-y,z})
\end{eqnarray*}
entail 
\[\mathsf{Advantage}(x,z,\mathbf{R}|_{\{x,z\}})\not > \mathsf{Standard}(x,z,\mathbf{R}^{-x,z}).\]
Generally speaking, what social rationality assumptions hold within the AS model depends on the features of the underlying advantage and standard functions. Not only is there no guarantee that a CCR is negatively transitive on the AS model, but also the following theorem of Baigent \citeyearpar{Baigent1987} shows that we cannot hope to escape variations of Arrow's Theorem within the AS model while keeping Arrow's assumption of negative transitivity. A \textit{weak dictator} (resp.~\textit{vetoer}) for a CCR $f$ is an $i\in V$ such that for any $x,y\in X$ and profile $\mathbf{R}$, if $xP(\mathbf{R}_i)y$, then $xf(\mathbf{R})y$ (resp.~\textit{not} $yP(f(\mathbf{R}))x$). 

\begin{theorem}[\citealt{Baigent1987}] Assume $|X|\geq 4$. If $f$ is an SWF satisfying weak IIA and Pareto, then there is a weak dictator for $f$.\footnote{Campbell and Kelly \citeyearpar{Campbell2000} observe that Baigent's original assumption of $|X|\geq 3$ must be strengthened to $|X|\geq 4$.}
\end{theorem}

\noindent Weakening the SWF condition to only assume negative transitivity, we have the following.

\begin{proposition} Assume $|X|\geq 4$. If $f$ is a CCR satisfying weak IIA and Pareto such that for every profile $\mathbf{R}$, $P(f(\mathbf{R}))$ is negatively transitive, then there is a vetoer for~$f$.
\end{proposition}

\begin{proof} The proof is the same as for Proposition \ref{StrongerArrow} only now observing that if $f'$ has a weak dictator, then $f$ has a vetoer.\end{proof}

Thus, we conclude that not only must IIA be weakened to AS representability, but also negative transitivity must be weakened. As we shall see, acyclicity and even transitivity are compatible with AS representability and Pareto.

\subsection{Covering}\label{CoveringSection}

In this section, we show that a CCR based on the \textit{covering relation} of \citealt{Gillies1959} is AS representable. Covering is based on the majority relation $\succ_\mathbf{R}$ for a preprofile $\mathbf{R}$ defined by 
 \[x\succ_{\mathbf{R}}y \Longleftrightarrow \mathrm{Margin}_\mathbf{R}(x,y)>0.\] 

\begin{definition} Given $x,y\in X$ and a profile $\mathbf{R}$, define:
\begin{enumerate}
    \item $xP_{cov}(\mathbf{R})y$ if $x\succ_\mathbf{R}y$ and for all $v\in X$,  $v\succ_\mathbf{R}x$ implies $v\succ_\mathbf{R}y$;
    \item $xI_{cov}(\mathbf{R})y$ if for all $v\in X$, $v\succ_\mathbf{R} x\Leftrightarrow v\succ_\mathbf{R} y$, and $x\succ_\mathbf{R} v\Leftrightarrow y\succ_\mathbf{R} v$.\footnote{Normally covering only concerns strict preference, not indifference, but we find this definition of indifference in the spirit of the covering definition for strict preference.}
\end{enumerate}
    The Gillies Covering CCR is defined by $xf_{cov}(\mathbf{R})y$ if  $xP_{cov}(\mathbf{R})y$ or $xI_{cov}(\mathbf{R})y$.
\end{definition}

\begin{remark} Miller \citeyearpar{Miller1980} independently introduced a slightly different covering relation: $x$ covers $y$ in Miller's sense if  $x\succ_\mathbf{R}y$ and for all $v\in X$,  $y\succ_\mathbf{R}v$ implies $x\succ_\mathbf{R}v$. Miller's covering  is equivalent to Gillies' if there is an odd number of voters with no ties in their rankings, but the two definitions are not equivalent in general. For other variations on the definition of covering, see \citealt{Duggan2013}, and for a notion of \textit{weighted} covering, see \citealt[Def.~3.2]{Dutta1999}. Miller's notion and that of weighted covering also lead to AS representable CCRs, but to fix ideas we focus on the Gillies notion.
\end{remark}

If the majority relation $\succ_\mathbf{R}$ contains no cycles (no sequence $x_1,\dots,x_n$ with $x_i\succ_\mathbf{R}x_{i+1}$ and $x_1=x_n$), then $P_{cov}$ is simply $\succ_\mathbf{R}$. But even if $\succ_\mathbf{R}$ contains cycles, $P_{cov}$ may be nonempty. To see this graph-theoretically, given a preprofile $\mathbf{R}$, we form a directed graph $M(\mathbf{R})$, the \textit{majority graph of $\mathbf{R}$}, with an edge from $a$ to $b$ whenever $a\succ_\mathbf{R}b$. Figure \ref{CoverExample} shows a majority graph of a profile---with two majority cycles---and the associated covering relation, which is nonempty. (Later we will encounter CCRs that can accept even more majority preferences in the presence of cycles. See, e.g., Figures \ref{RPfig} and \ref{SCexample}.)

\begin{figure}
\begin{center}\begin{tikzpicture}

\node[circle,draw, minimum width=0.25in] at (0,0) (a) {$a$}; 
\node[circle,draw,minimum width=0.25in] at (3,0) (c) {$c$}; 
\node[circle,draw,minimum width=0.25in] at (1.5,1.5) (b) {$b$}; 
\node[circle,draw,minimum width=0.25in] at (1.5,-1.5) (d) {$d$}; 

\path[->,draw,thick] (a) to (b);
\path[->,draw,thick] (b) to (c);
\path[->,draw,thick] (c) to  (d);
\path[->,draw,thick] (d) to  (a);
\path[->,draw,thick] (d) to  (b);
\path[->,draw,thick] (a) to  (c);

  \end{tikzpicture}\qquad\quad\begin{tikzpicture}

\node[circle,draw, minimum width=0.25in] at (0,0) (a) {$a$}; 
\node[circle,draw,minimum width=0.25in] at (3,0) (c) {$c$}; 
\node[circle,draw,minimum width=0.25in] at (1.5,1.5) (b) {$b$}; 
\node[circle,draw,minimum width=0.25in] at (1.5,-1.5) (d) {$d$}; 

\path[->,draw,thick] (a) to node[fill=white] {$P$} (b);

  \end{tikzpicture}\end{center}
    \caption{a majority graph (left) with cycles $a\to b\to c\to d\to a$ and $a\to c\to d\to a$, together with its associated covering relation (right): $a$ covers $b$ because not only is $a$ majority preferred to $b$, but the only candidate majority preferred to $a$, namely $d$, is also majority preferred to $b$.}
    \label{CoverExample}
\end{figure}
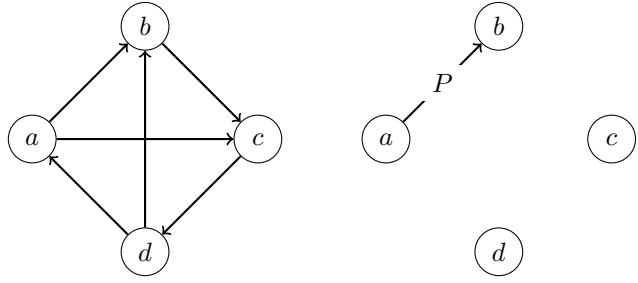

The Gillies Covering CCR clearly satisfies anonymity, neutrality, strong Pareto, Pareto indifference, and has no vetoers. Moreover, it is a transitive CCR.

\begin{restatable}{proposition}{GilliesTrans} The Gillies Covering CCR is transitive.
\end{restatable}

In addition, we now observe that the Gillies Covering CCR is AS representable. 

\begin{restatable}{proposition}{GilliesAS}\label{GilliesAS} The Gillies Covering CCR is AS representable with
\[\mathsf{Advantage}(x,y,\mathbf{R}|_{\{x,y\}}):=
\mathrm{Margin}_{\mathbf{R}|_{\{x,y\}}}(x,y)\]
and 
\[\mathsf{Standard}(x,y,\mathbf{R}^{-x,y}):=\begin{cases}
0\mbox{ if for all $v\in X\setminus\{x,y\}$, $v\succ_{\mathbf{R}^{-x,y}}x$ implies $v\succ_{\mathbf{R}^{-x,y}}y$}\\
|V| \mbox{ otherwise}.
\end{cases}\]
\end{restatable}

Thus, weakening IIA to AS representability makes room for CCRs that satisfy the rest of Arrow's axioms---except for completeness. For many profiles, $f_{cov}(\mathbf{R})$ is not a complete relation: there may be $x,y\in X$ such that neither $xf_{cov}(\mathbf{R})y$ nor $yf_{cov}(\mathbf{R})x$, as in Figure~\ref{CoverExample}. We do not regard this as a reason for rejecting the Gillies Covering CCR. Given an incomplete relation $R$, one can still perfectly well induce a choice function on $X$ by selecting \textit{maximal elements}: for any nonempty $Y\subseteq X$,
\[M(Y,R)=\{x\in Y\mid \mbox{there is no }y\in Y: yP(R)x\}.\]
For a study of maximal element choice, see \citealt{Bossert2005}. As the failure of completeness arguably makes sense even for individuals (see, e.g.,   \citealt{Aumann1962}, \citealt{Putnam1986,Putnam2002}, \citealt{Chang1997}, and \citealt{Eliaz2006}), this provides still more reason to doubt the requirement of completeness at the social level. For further doubts about completeness of social preference, see \citealt{Sen2018}.

Arrow's \citeyearpar[p.~118]{Arrow1963} argument for completeness is based on the failure to consider maximal element choice, considering only greatest element choice: for nonempty $Y\subseteq X$,
\[G(Y,R)=\{x\in Y\mid \mbox{for all  }y\in Y,\, x f(R)y\}.\]
Arrow claims that completeness ``when understood, can hardly be denied; it simply requires that some social choice be made from any environment. Absention from a decision cannot exist; some social state will prevail'' (p.~118). Indeed, the condition that $G(Y,R)\neq\varnothing$ for all nonempty $Y\subseteq X$ does imply that $R$ is complete. But the condition that $M(Y,R)\neq\varnothing$ for all nonempty $Y\subseteq X$ does not. Another crucial point is that even with a complete relation $R$ and greatest element choice, $G(Y,R)$ may contain multiple ``tied'' alternatives; thus, completeness does not confer a special guarantee that $G(Y,R)$ will be a singleton set.

One might also try to argue for completeness on the grounds that given any incomplete $R$, it is innocuous to switch to a complete $R'$ defined by $xR'y$ if \textit{not} $yP(R)x$. But this is not innocuous, for two reasons. First, while the transformation from $R$ to $R'$ preserves transitivity of the strict relation, given that $P(R)=P(R')$, it does not preserve the transitivity of the weak relation: for $R$ transitive, $R'$ may not be transitive. Second, although moving from $R$ to $R'$ does not matter for the induced maximal element choice functions, as $P(R)=P(R')$ implies $M(\cdot,R)=M(\cdot,R')$,  one may care about differences between $R$ and $R'$ even when $M(\cdot,R)=M(\cdot,R')$, because such differences between $R$ and $R'$ may affect later stages of the choice process, of which the reduction from $Y$ to $M(Y,R)$ may be only the first stage (see \citealt[pp.~14-15]{Schwartz1986}). In particular, for some $Z\subseteq Y$ with  $|Z|\geq 2$, we may have $M(Y,R)=Z$ and $M(Y,R')=Z$ for different reasons: it may be that according to $R$, the elements of $Z$ are \textit{noncomparable} for society, so none can be eliminated at this stage; by contrast, according to $R'$, society is \textit{indifferent} between the elements of $Z$, so again none can be eliminated but also there is a positive judgment that the options are equally preferable. This difference between noncomparability and indifference may well have consequence for later stages of the social choice process, e.g., whether the process leads to further deliberation or to a random tiebreaking mechanism. 

When the social relation is possibly incomplete,  we can distinguish the following two strengthenings of weak IIA, one of which Gillies Covering satisfies.

\begin{definition}\label{PNPIdef} Let $f$ be a CCR.
\begin{enumerate}
    \item $f$ satisfies\textit{ PN-weak IIA} if for all $x,y\in X$ and  profiles $\mathbf{R},\mathbf{R}'$ with $\mathbf{R}|_{\{x,y\}}=\mathbf{R}'|_{\{x,y\}}$:
\[xP(f(\mathbf{R}))y\implies xP(f(\mathbf{R}'))y \mbox{ or }xN(f(\mathbf{R}'))y.\]
\item $f$ satisfies\textit{ PI-weak IIA} if for all $x,y\in X$ and profiles $\mathbf{R},\mathbf{R}'$ with  $\mathbf{R}|_{\{x,y\}}=\mathbf{R}'|_{\{x,y\}}$:
\[xP(f(\mathbf{R}))y\implies xP(f(\mathbf{R}'))y \mbox{ or }xI(f(\mathbf{R}'))y.\]
\end{enumerate}
\end{definition}

\begin{restatable}{proposition}{GilliesPN}\label{GilliesPN} The Gillies Covering CCR satisfies PN-weak IIA.
\end{restatable}

It is no accident that the Gillies Covering CCR satisfies PN-weak IIA rather than PI-weak IIA, as the following impossibility theorem shows. A coalition $C\subseteq V$ is \textit{weakly decisive} if for every profile $\mathbf{R}$ in which $xP(\mathbf{R}_i)y$ for all $i\in C$, we have $xf(\mathbf{R})y$. The following result may be compared to Weymark's \citeyearpar{Weymark1984} Oligarchy Theorem assuming transitivity and IIA.

\begin{restatable}{proposition}{OligarchyTheorem}\label{OligarchyTheorem} Let $|X|\geq 4$ and $V$ be finite. If $f$ is a transitive CCR satisfying PI-weak IIA and Pareto, then there is a nonempty $C\subseteq V$ that is weakly decisive and such that every $i\in C$ is a vetoer.
\end{restatable}

\begin{figure}

\begin{center}
\begin{minipage}{1.25in}\begin{tabular}{ccc}
$5$ & $2$ & $3$   \\\hline
$a$ & $b$ &  $c$ \\
$b$ &  $c$ & $a$ \\
$c$ &  $a$ &  $b$ \\
\end{tabular}\end{minipage}\begin{minipage}{4in}\begin{tikzpicture}

\node[circle,draw, minimum width=0.25in] at (0,0) (a) {$a$}; 
\node[circle,draw,minimum width=0.25in] at (3,0) (c) {$c$}; 
\node[circle,draw,minimum width=0.25in] at (1.5,1.5) (b) {$b$}; 

\path[->,draw,thick] (b) to  (c);
\path[->,draw,thick] (a) to (b);

\end{tikzpicture}\qquad\quad\begin{tikzpicture}

\node[circle,draw, minimum width=0.25in] at (0,0) (a) {$a$}; 
\node[circle,draw,minimum width=0.25in] at (3,0) (c) {$c$}; 
\node[circle,draw,minimum width=0.25in] at (1.5,1.5) (b) {$b$}; 

\path[->,draw,thick] (a) to node[fill=white] {$P$} (b);

\end{tikzpicture}
\end{minipage}
\end{center}
    \caption{an example in which $P(f_{cov}(\mathbf{R}))$ is not negatively transitive. The first graph is the majority graph for the profile $\mathbf{R}$ on the left, showing that $a\succ_\mathbf{R}b$ and $b\succ_\mathbf{R}c$, but neither $a\succ_\mathbf{R} c$ nor $c\succ_\mathbf{R} a$. The second graph shows the induced covering relation: since $a\succ_\mathbf{R}b$ but $a\not\succ_\mathbf{R}c$, it is not the case that $bP(f_{cov}(\mathbf{R}))c$. Finally, since \textit{not} $aP(f_{cov}(\mathbf{R}))c$, \textit{not} $cP(f_{cov}(\mathbf{R}))b$, and yet $aP(f_{cov}(\mathbf{R}))b$, $P(f_{cov}(\mathbf{R}))$ is not negatively transitive.}
    \label{NegativeTrans}
\end{figure}
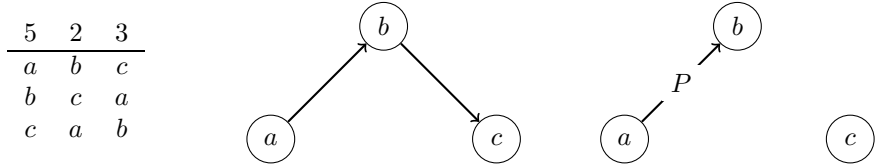

Finally, we return to the point that $P(f_{cov}(\mathbf{R}))$ can fail to be negatively transitive, as in Figure \ref{NegativeTrans}, as we expected from Baigent's theorem. As a result, the maximal element choice function induced by $f_{cov}$ does not satisfy as strong ``choice consistency'' conditions as that induced by a relation $R$ for which $P(R)$ is negatively transitive. In particular, $M(\cdot, f_{cov}(\mathbf{R}))$ can violate Sen's \citeyearpar{Sen1969a} $\beta$ condition:\footnote{By contrast, $M(\cdot, R)$ satisfies $\beta$ whenever $P(R)$ is negatively transitive.} \[\mbox{if $\varnothing\neq Y\subseteq Z\subseteq X$ and $C(Y)\cap C(Z)\neq \varnothing$, then $C(Y)\subseteq  C(Z)$.}\]
For example, for the profile in Figure \ref{NegativeTrans}, we have the following violation of $\beta$:
\begin{eqnarray*}
M(\{b,c\},f_{cov}(\mathbf{R}))=\{b,c\};\\ M(\{a,b,c\}, f_{cov}(\mathbf{R}))=\{a,c\}.
\end{eqnarray*}
Yet in his argument for ``social rationality'' conditions, Arrow \citeyearpar{Arrow1963} does not give an argument that a choice function for society ought to satisfy $\beta$. Instead, he gives an argument (p.~120) that a choice function for society ought to satisfy \textit{path independence}, which can be formalized as the condition that for all $Y_1,Y_2\subseteq X$,
\[C(Y_1\cup Y_2)=C(C(Y_1)\cup C(Y_2)).\]
As observed by Plott \citeyearpar{Plott1973}, if the strict relation $P(R)$ is transitive---in which case $R$ is said to be \textit{quasi-transitive}---then $M(\cdot,R)$ satisfies path independence.\footnote{See Theorem 1 of \citealt{Blair1976} for a conjunction of conditions equivalent to path independence. Plott \citeyearpar{Plott1973} characterizes choice functions coming from quasi-transitive relations using path independence and a second condition, known as the Generalized Condorcet Axiom: for all $x\in Y\subseteq X$, if $x\in C(\{x,y\})$ for all $y\in Y$, then $x\in C(Y)$. We note that the following are equivalent conditions on a choice function $C$: $C=G(\cdot, R)$ for some quasi-transitive and complete relation $R$; $C=M(\cdot,R)$ for some quasi-transitive and complete relation $R$; $C=M(\cdot,R)$ for some transitive and reflexive relation $R$; $C=M(\cdot,R)$ for some transitive and regular (see \citealt[Theorem 2]{Eliaz2006}) relation $R$. Additional equivalences may be found in \citealt[Theorem 2]{Schwartz1976}.} Negative transitivity is not required. As $P(f_{cov}(\mathbf{R}))$ is transitive, we conclude that there are AS representable CCRs that induce maximal element choice functions satisfying Arrow's desideratum of path independence. Thus, weakening IIA to AS representability allows us to give Arrow the social choice consistency condition he desired, without dictators or vetoers.

In the next two subsections, we consider CCRs that become available if we do not insist on path independence of the choice function induced by social preference. As Plott \citeyearpar[p.~1090]{Plott1973} remarks, ``[T]here appears to be no overriding reason to impose even I.P. [independence of path] at the very outset of the analysis, even though special considerations make I.P. appear to be a potentially useful tool.'' One specific objection to path independence is that it ignores the holistic nature of fairness. For example, consider an election with the following features, depicted in Figure~\ref{PathIndCounter}: 3 more voters prefer $a$ over $b$ than prefer $b$ over $a$; 3 more voters prefer $b$ over $c$ than prefer $c$ over $b$; and $1$ more voter prefers $c$ over $a$ than prefers $a$ over $c$. If all three candidates $a,b,c$ are still feasible---none has died, dropped out after election day, etc.---then it would be \textit{unfair to $b$ and $b$'s supporters} to choose $c$ but not $b$ as a winner from $\{a,b,c\}$, as $b$'s global position is at least as strong as $c$'s in Figure \ref{PathIndCounter}. Thus, it should be that (i) if $c\in C(\{a,b,c\})$, then  $b\in C(\{a,b,c\})$. However, if $b$ were to die or drop out, so the feasible set of candidates shrinks from $\{a,b,c\}$ to $\{a,c\}$, then there is no unfairness to $b$ or $b$'s supporters in choosing $c$ as a winner from $\{a,c\}$. Indeed, we should have (ii) $c\in C(\{a,c\})$. There is also no unfairness to anyone if we choose $a$ as the unique winner in the case that $c$ (rather than $b$) were to die or drop out, so we should be able to have (iii) $C(\{a,b\})=\{a\}$. But together (i), (ii), and (iii) contradict path independence. For path independence requires \[C(\{a,b,c\})=C(C(\{a,b\})\cup C(\{c\})).\] By (iii), we can rewrite the right-hand side as $C(\{a\}\cup C(\{c\}))=C(\{a,c\})$, which by (ii) contains $c$, but it does not contain $b$ since $C(\{a,c\})\subseteq \{a,c\}$. Thus, by path independence, $C(\{a,b,c\})$ contains $c$ but not $b$, violating the fairness condition in (i).

The moral of the example above is that we should not expect that making the choice with \textit{all three candidates $a,b,c$ still in the running}, which requires fairness to each of the three candidates and their supporters, is the same as first making the choice between $a,b$, without fairness considerations for $c$ (who we imagine has dropped out of contention), and then making a choice between the winner(s) of that contest and $c$ (who now we imagine is back in contention after all), which may now involve ignoring fairness considerations to $b$ (if $a$ was chosen over $b$ from $a,b$). Path independence has us ignoring fairness considerations by the way we artificially split up the decision. For this reason, we do not wish to restrict attention only to CCRs whose induced choice functions satisfy path independence. 

\begin{figure}
\begin{center} 
\begin{tikzpicture}

\node[circle,draw, minimum width=0.25in] at (0,0) (a) {$a$}; 
\node[circle,draw,minimum width=0.25in] at (3,0) (c) {$c$}; 
\node[circle,draw,minimum width=0.25in] at (1.5,1.5) (b) {$b$}; 

\path[->,draw,thick] (c) to node[fill=white] {$1$} (a);
\path[->,draw,thick] (b) to node[fill=white] {$3$} (c);
\path[->,draw,thick] (a) to node[fill=white] {$3$} (b);

\end{tikzpicture}
\end{center}
\caption{a margin graph. The arrow from $a$ to $b$ with weight 3 indicates that 3 more voters prefer $a$ to $b$ than prefer $b$ to $a$, etc.}\label{PathIndCounter}
\end{figure}
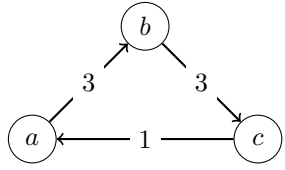
 
\subsection{Ranked Pairs}\label{RankedPairsSection}

In this section, we show that a CCR based on Tideman's \citeyearpar{Tideman1987} Ranked Pairs voting method is AS~representable. Ranked Pairs belongs to the family of methods that are sensitive to the strength of majority preference of one candidate over another. The key idea is that in a profile where $a$ is majority preferred to $b$, and $c$ is majority preferred to $d$, it may be that the majority preference for $a$ over $b$ is \textit{stronger} than the majority preference for $c$ over $d$: $(a,b)S_\mathbf{R}(c,d)$. There are a number of possible ways to define $S_{\mathbf{R}}$---for instance, defining $S_{\mathbf{R}}$ in terms of $\mathrm{Margin}_{\mathbf{R}}$ or $\mathrm{Ratio}_{\mathbf{R}}$. For ease of exposition, we fix our measure of strength of majority preference to be $\mathrm{Margin}_{\mathbf{R}}$ (though see Footnote \ref{RPRatioPareto} and Remark \ref{RPratio}). 

Any preprofile $\mathbf{R}$ induces a weighted directed graph $\mathcal{M}(\mathbf{R})$, the \textit{margin graph of $\mathbf{R}$}, whose set of vertices is $X$ with an edge from $x$ to $y$ when $\mathrm{Margin}_{\mathbf{R}}(x,y)> \mathrm{Margin}_{\mathbf{R}}(y,x)$,\footnote{The condition that $\mathrm{Margin}_{\mathbf{R}}(x,y)> \mathrm{Margin}_{\mathbf{R}}(y,x)$ is equivalent to $\mathrm{Margin}_{\mathbf{R}}(x,y)> 0$, but if one wishes to switch $Margin$ to $Ratio$ (see Remark \ref{RPratio}), we need $\mathrm{Ratio}_{\mathbf{R}}(x,y)> \mathrm{Ratio}_{\mathbf{R}}(y,x)$ rather than $\mathrm{Ratio}_{\mathbf{R}}(x,y)>0$.} weighted by $\mathrm{Margin}_{\mathbf{R}}(x,y)$.  Ranked Pairs has the property that for any $\mathbf{R}$ and $\mathbf{R}'$, if $\mathcal{M}(\mathbf{R})=\mathcal{M}(\mathbf{R}')$, then the output of Ranked Pairs is the same for $\mathbf{R}$ and $\mathbf{R}'$.  We may even think of Ranked Pairs as taking as input any weighted directed graph with positive weights, which we call a \textit{margin graph} $\mathcal{M}$. For an edge $\langle x,y\rangle$ in $\mathcal{M}$, let $\mathrm{Margin}_\mathcal{M}(x,y)$ be its weight.

Roughly speaking, Ranked Pairs locks in majority preference relations in order of strength, ignoring any majority preferences that would create a cycle with majority preferences already locked in. As Tideman \citeyearpar[p.~199]{Tideman1987} informally describes it:
\begin{quote} Start with the pairings decided by the largest and second largest majorities, and require that the orderings they specify be preserved in the final ranking of all candidates. Seek next to preserve the pair-ordering decided by the third largest majority, and so on. When a pair-ordering is encountered that cannot [without creating a cycle] be preserved while also preserving all pair-orderings with greater majorities, disregard it and go on to pair-orderings decided by smaller majorities. 
\end{quote}
For example, for the margin graph in Figure \ref{RPfig}, we first lock in the margin 5 victory of $a$ over $b$, then lock in the margin 3 victory of $b$ over $c$, and then ignore the margin $1$ victory of $c$ over $a$, as it is inconsistent with the relations already locked in. Note, by contrast, that $a,b,c$ are noncomparable according to the Gillies Covering~CCR.

\begin{figure}[h]
\begin{center}
\begin{minipage}{1.25in}\begin{tabular}{ccc}
$4$ & $2$ & $3$   \\\hline
$a$ & $b$ &  $c$ \\
$b$ &  $c$ & $a$ \\
$c$ &  $a$ &  $b$ \\
\end{tabular}\end{minipage}\begin{minipage}{4in}\begin{tikzpicture}

\node[circle,draw, minimum width=0.25in] at (0,0) (a) {$a$}; 
\node[circle,draw,minimum width=0.25in] at (3,0) (c) {$c$}; 
\node[circle,draw,minimum width=0.25in] at (1.5,1.5) (b) {$b$}; 

\path[->,draw,thick] (b) to node[fill=white] {$3$} (c);
\path[->,draw,thick] (c) to node[fill=white] {$1$} (a);
\path[->,draw,thick] (a) to node[fill=white] {$5$} (b);

\end{tikzpicture}\qquad\quad\begin{tikzpicture}

\node[circle,draw, minimum width=0.25in] at (0,0) (a) {$a$}; 
\node[circle,draw,minimum width=0.25in] at (3,0) (c) {$c$}; 
\node[circle,draw,minimum width=0.25in] at (1.5,1.5) (b) {$b$}; 

\path[->,draw,thick] (b) to node[fill=white] {$P$} (c);
\path[->,draw,thick] (a) to node[fill=white] {$P$} (b);

\end{tikzpicture}
\end{minipage}\end{center}
\caption{an illustration of Ranked Pairs. The first graph is the margin graph for the profile $\mathbf{R}$ on the left. The second graph shows the majority preferences locked in by Ranked Pairs.}
    \label{RPfig}
\end{figure}
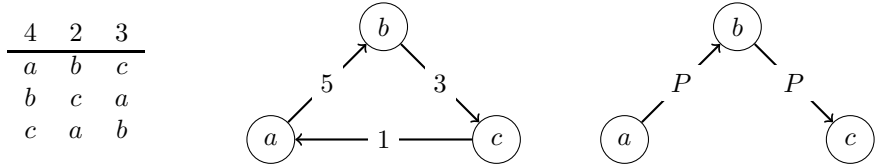

Formally, let  $\mathcal{M}$ be a margin graph and $T\in \mathcal{L}(X^2\setminus \Delta_X)$ a linear order on the set of pairs of distinct candidates (here $\Delta_X=\{\langle x,x\rangle\mid x\in X\}$), which we will use to break ties when we have distinct pairs $\langle x,y\rangle$ and $\langle x',y'\rangle$ such that $\mathrm{Margin}_\mathcal{M}(x,y)=\mathrm{Margin}_\mathcal{M}(x',y')$. Define  $\langle x,y\rangle >_{\mathcal{M},T}\langle x',y'\rangle$ if either  $\mathrm{Margin}_\mathcal{M}(x,y) > \mathrm{Margin}_\mathcal{M}(x',y')$ or $\mathrm{Margin}_\mathcal{M}(x,y) = \mathrm{Margin}_\mathcal{M}(x',y')$ and $\langle x,y\rangle \mathrel{T} \langle x',y'\rangle$.  Then $\langle x,y\rangle >_{\mathcal{M},T}\langle x',y'\rangle$ means that we consider the pair $\langle x,y\rangle$ before the pair $\langle x',y'\rangle$, as in Tideman's description. Let \[E_\mathcal{M}=\{\langle x,y\rangle\mid \mathrm{Margin}_\mathcal{M}(x,y)>\mathrm{Margin}_\mathcal{M}(y,x) \}\]
be the edge relation in the margin graph $\mathcal{M}$. We define the \textit{Ranked Pairs relation} $\mathbb{RP}(\mathcal{M},T)$, a subset of $E_\mathcal{M}$, inductively as follows ($C_n$ is the set of edges that have already been considered after stage $n$):
\begin{itemize}
\item $\mathbb{RP}(\mathcal{M},T)_0\,=\varnothing$ and $C_0=\varnothing$.
\item If $n<|E_\mathcal{M}|$, let $\langle a,b\rangle$ be the maximum element of $E_\mathcal{M} \setminus C_n$ according to $>_{\mathcal{M},T}$. Then define 
\[\mathbb{RP}(\mathcal{M},T)_{n+1} =\begin{cases} \mathbb{RP}(\mathcal{M},T)_{n}\cup\{\langle a,b\rangle\}&\mbox{if this relation is acyclic} \\ \mathbb{RP}(\mathcal{M},T)_{n} & \mbox{otherwise}.  \end{cases}\]
In either case, let $C_{n+1}=C_n\cup \{\langle a,b\rangle\}$.
\item $\mathbb{RP}(\mathcal{M},T)=\mathbb{RP}(\mathcal{M},T)_{|E_\mathcal{M}|}$.
\end{itemize}

Since $\mathbb{RP}(\mathcal{M},T)$ is an acyclic relation on $X$ by construction, one can take its transitive closure $\mathbb{RP}(\mathcal{M},T)^*$, as Tideman does when discussing the ranking determined by Ranked Pairs, and $\mathbb{RP}(\mathcal{M},T)^*$ is also acyclic.\footnote{Note that the maximal elements of $\mathbb{RP}(\mathcal{M},T)$ and $\mathbb{RP}(\mathcal{M},T)^*$ are the same, i.e.,
\[\{x\in X\mid \mbox{there is no }y: \langle y,x\rangle\in \mathbb{RP}(\mathcal{M},T)\}=\{x\in X\mid \mbox{there is no }y: \langle y,x\rangle \in \mathbb{RP}(\mathcal{M},T)^* \}, \]
so they both determine the same set of winners for the election. This is a general fact: given an acyclic relation $P$ and its transitive closure $P^*$, the set of maximal elements of $P$ is equal to the set of maximal elements of $P^*$.  But differences appear if we consider the maximal element choice functions from Section \ref{CoveringSection}, induced by the two relations, applied to proper subsets of the set of all candidates. For example, compare $M(\cdot, \mathbb{RP}(M,T))$ and $M(\cdot, \mathbb{RP}(M,T)^*)$ on $\{a,c\}$ in Figure \ref{RPfig}: $M(\{a,c\}, \mathbb{RP}(\mathcal{M},T)^*)=\{a\}$, despite that $c$ is majority preferred to $a$, whereas $M(\{a,c\}, \mathbb{RP}(\mathcal{M},T))=\{a,c\}$, suggesting a further tie-breaking procedure.} However, we will not use $\mathbb{RP}(\mathcal{M},T)^*$ in what follows because doing so would result in a CCR that violates weak IIA and hence AS representability. For example, in Figure \ref{RPfig}, the transitive closure adds a strict social preference of $a$ over $c$, \textit{even though $c$ is majority preferred to $a$}. To turn this into a violation of weak IIA, consider the modified profile where the two voters in the middle column move $c$ to the top of their ranking; then $c$ is strictly socially preferred to $a$ according to both $\mathbb{RP}(\mathcal{M},T)$ and $\mathbb{RP}(\mathcal{M},T)^*$. Thus, using $\mathbb{RP}(\mathcal{M},T)^*$ would yield reversed strict social preferences on $a$ vs.~$c$ for two profiles in which all voters rank $a$ vs.~$c$ in the same way, violating weak IIA. This explains our choice of $\mathbb{RP}(\mathcal{M},T)$ for an AS representable version of Ranked Pairs.

In order to have a neutral version of Ranked Pairs, we eliminate the dependence on $T$ by defining the relation $\mathbb{RP}(\mathcal{M})$ by 
\[\langle x,y\rangle\in \mathbb{RP}(\mathcal{M})\mbox{ if for all }T\in\mathcal{L}(X^2\setminus \Delta_X),\langle x,y\rangle\in \mathbb{RP}(\mathcal{M},T).\]
Thus, we keep a strict preference for $x$ over $y$ if and only if that preference gets locked in by the Ranked Pairs procedure no matter the tie-breaking ordering $T$.

The relation $\mathbb{RP}(\mathcal{M})$ is a \textit{strict} social preference relation. There are then various options for defining the weak social preference relation outputted by a CCR, including anonymous and neutral options.  For example, we could define $f$ such that $xf(\mathbf{R})y$ if $\langle y,x\rangle\not\in \mathbb{RP}(\mathcal{M}(\mathbf{R}))$. This is a complete CCR, but there are also incomplete options. For example, define $xf(\mathbf{R})y$ if either $\langle x,y\rangle\in \mathbb{RP}(\mathcal{M}(\mathbf{R}))$ or all voters are indifferent between $x$ and $y$, thereby satisfying Pareto indifference. Given the multiplicity of options, we will define a family of Ranked Pairs CCRs instead of a unique Ranked Pairs CCR.

\begin{definition} A Ranked Pairs CCR is a CCR $f$ such that for any profile $\mathbf{R}$, $xP(f(\mathbf{R}))y$ if and only if $\langle x,y\rangle \in \mathbb{RP}(\mathcal{M}(\mathbf{R}))$.
\end{definition}

 Every Ranked Pairs CCR $f$ satisfies acyclicity and Pareto\footnote{Using $Ratio$ instead of $Margin$, strong Pareto is also satisfied.\label{RPRatioPareto}} and has no vetoers. Further properties of Ranked Pairs are described in \citealt{Tideman1987}.

We now observe that Ranked Pairs CCRs are AS representable. Intuitively, the advantage according to Ranked Pairs is the margin of victory; and the standard for social preference of $x$ over $y$ in $\mathbf{R}$ is the greatest margin of victory that $x$ can have over $y$ according to a margin graph $\mathcal{M}$ that matches $\mathcal{M}(\mathbf{R})$ except on $x,y$ and yet $\left\langle x,y\right\rangle \notin \mathbb{RP}(\mathcal{M})$.

\begin{restatable}{proposition}{RankedPairsAS}\label{RankedPairsAS} Any Ranked Pairs CCR  is AS representable with
\[\mathsf{Advantage}(x,y,\mathbf{R}|_{\{x,y\}}):=\mathrm{Margin}_{\mathbf{R}|_{\{x,y\}}}(x,y)\]
and
\[\mathsf{Standard}(x,y,\mathbf{R}^{-x,y})= \mathrm{min}\{k-1 \mid k\in \mathbb{Z}^+,\, \langle x,y\rangle\in \mathbb{RP}(\mathcal{M}(\mathbf{R}^{-x,y})+x\overset{k}{\to}y)\}\]
where $\mathcal{M}(\mathbf{R}^{-x,y})+x\overset{k}{\to}y$ is the margin graph obtained from $\mathcal{M}(\mathbf{R}^{-x,y})$ by adding an edge from $x$ to $y$ with weight $k$.
\end{restatable}

\begin{remark}\label{RPratio}
Defining Ranked Pairs as above but with \textit{Ratio} instead of \textit{Margin}\footnote{Thus, the graphs are  weighted by $Ratio$ rather than $Margin$ in the algorithm for Ranked Pairs.} produces an AS representable CCR, where the advantage and standard functions---taking values in $\mathbb{Q}_{>0}$ with its usual multiplication operation and ordering---are as follows: $\mathsf{Advantage}$ is as in Proposition \ref{RankedPairsAS} but with \textit{Ratio} instead of \textit{Margin}; letting $\mathcal{R}_{|V|}$ be the finite set of possible values of \textit{Ratio} (for a fixed number $|V|$ of voters) and $m_{|V|}=\min\{|x-y|:x,y\in \mathcal{R}_{|V|}\}$, \[\mathsf{Standard}(x,y,\mathbf{R}^{-x,y})=\min\{k-m_{|V|}\mid k\in \mathcal{R}_{|V|},\langle x,y\rangle\in \mathbb{RP}(\mathcal{M}(\mathbf{R}^{-x,y})+x\overset{k}{\to}y)\}.\]
\end{remark}

\subsection{Split Cycle}

In this section, we show that the Split Cycle CCR studied in \citealt{HP2020a,HP2020b} is AS representable. The Split Cycle CCR outputs a strict relation of \textit{defeat} between candidates, corresponding to our strict social preference $P(f(\mathbf{R}))$, without addressing questions of indifference between candidates. Like Ranked Pairs, Split Cycle decides strict social preferences using the majority margins between candidates in a profile. In this paper, where we allow ties in ballots, we can consider other notions of strength of majority preference, such as the   $\mathrm{Ratio}_\mathbf{R}$ measure defined in Section \ref{ASmodelsection}.  For ease of exposition, however, we will continue to use $\mathrm{Margin}_{\mathbf{R}}$ (though see Footnote \ref{SCRatioPareto} and Remark \ref{SCratio}).

The Split Cycle defeat relation can be determined as follows:
\begin{enumerate}
    \item In each majority cycle, identify the majority preference with the smallest margin in that cycle (e.g., if $a$ beats $b$ by 5, $b$ beats $c$ by 3, and $c$ beats $a$ by $1$, then the majority preference for $c$ over $a$ has the smallest margin in the cycle).
    \item After completing 1 for all cycles, discard the identified majority preferences. All remaining majority preferences count as defeats.
\end{enumerate}
Thus, $a$ defeats $b$ just in case there is no majority cycle containing $a,b$ in which the majority preference for $a$ over $b$ is the weakest majority preference occurring in the cycle.

Figure \ref{SCexample} shows a margin graph and the associated strict social preference relation according to Split Cycle. The weakest majority preference in the  cycle $a\to c\to b\to a$ is that of $b\to a$ with a margin of 3, and the weakest majority preference in the cycles $a\to d\to b\to a$ and $a\to d\to c\to b\to a$ is that of $a\to d$ with a margin of $1$. Note that Ranked Pairs would lock in an additional strict preference $aPd$,\footnote{Ranked Pairs locks in $dPc$,  $aPc$,  $cPb$, and $dPa$, then ignores the majority preference for $b$ over $a$ as inconsistent with what has already been locked in, and finally locks in $aPd$} while the only strict preference produced by the Gillies Covering CCR is $dPc$.

In order to prove that Split Cycle is AS representable, we introduce some terminology and equivalent reformulations of Split Cycle from \citealt{HP2020a}.

\begin{definition} Given a preprofile $\mathbf{R}$ and $x,y\in X$, a \textit{majority path from $x$ to $y$} in $\mathbf{R}$ is a sequence $\rho = \langle z_1,\dots,z_n \rangle$ of candidates with $x=z_1$ and $y=z_n$ such that for each $i<n$, $z_i\succ_\mathbf{R}z_{i+1}$, and for $1\leq i < j <n$,  we have $z_i\neq z_j$ and $z_j\neq z_n$. The \textit{splitting number of $\rho$} is the smallest majority margin between consecutive candidates in $\rho$:
\[\mathrm{Split}\#_\mathbf{R}(\rho)= \mbox{min}\{\mathrm{Margin}_\mathbf{R}(z_i,z_{i+1})\mid i<n \}.\]
A \textit{majority cycle} is a majority path as above for which $z_1=z_n$.
\end{definition}

Since Split Cycle is defined in terms of a strict relation of defeat, without concern for a notion of indifference between candidates, we will follow the approach used for Ranked Pairs above and define a family of Split Cycle CCRs outputting a weak relation $f(\mathbf{R})$ rather than \textit{the} Split Cycle CCR outputting a weak relation $f(\mathbf{R})$; but all Split Cycle CCRs agree on the strict relation between candidates.

\begin{definition}\label{SplitCycleDef} A Split Cycle CCR is a CCR $f$ such that for any profile $\mathbf{R}$ and $x,y\in X$, we have $xP(f(\mathbf{R}))y$ if and only if $\mathrm{Margin}_\mathbf{R}(x,y)>\mathrm{Margin}_\mathbf{R}(y,x)$ and
\[\mathrm{Margin}_\mathbf{R}(x,y)>Split\#(\rho)\mbox{ for every majority cycle }\rho\mbox{ containing $x$ and $y$}.\]
\end{definition}

Every Split Cycle CCR satisfies acyclicity and Pareto\footnote{Using $Ratio$ in place of $Margin$, strong Pareto is also satisfied.\label{SCRatioPareto}} and has no vetoers; many other properties of Split Cycle are discussed in \citealt{HP2020b,HP2020a} and \citealt{Ding2022}. Note in particular that Split Cycle satisfies a much stronger axiom than weak IIA:  the Coherent~IIA axiom of \citealt{HP2020b}.

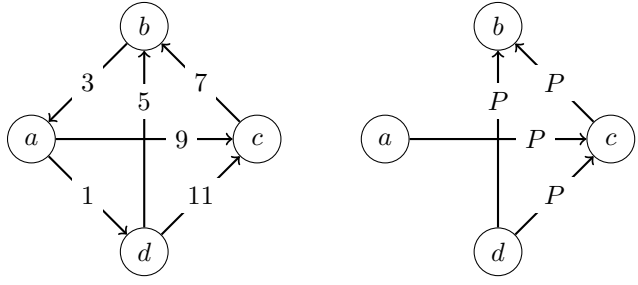
\begin{figure}
\begin{center}
\begin{minipage}{4in}\begin{tikzpicture}

\node[circle,draw, minimum width=0.25in] at (0,0) (b) {$a$}; 
\node[circle,draw,minimum width=0.25in] at (3,0) (a) {$c$}; 
\node[circle,draw,minimum width=0.25in] at (1.5,1.5) (c) {$b$}; 
\node[circle,draw,minimum width=0.25in] at (1.5,-1.5) (d) {$d$}; 

\path[->,draw,thick] (b) to (a);
\path[->,draw,thick] (d) to (c);
\path[->,draw,thick] (c) to node[fill=white] {$3$} (b);
\path[->,draw,thick] (a) to node[fill=white] {$7$} (c);
\path[->,draw,thick] (d) to node[fill=white] {$11$} (a);
\path[->,draw,thick] (b) to node[fill=white] {$1$} (d);

\node[fill=white] at (1.5,.5)  {$5$}; 
\node[fill=white] at (2,0)  {$9$}; 

  \end{tikzpicture}\qquad\quad\begin{tikzpicture}

\node[circle,draw, minimum width=0.25in] at (0,0) (b) {$a$}; 
\node[circle,draw,minimum width=0.25in] at (3,0) (a) {$c$}; 
\node[circle,draw,minimum width=0.25in] at (1.5,1.5) (c) {$b$}; 
\node[circle,draw,minimum width=0.25in] at (1.5,-1.5) (d) {$d$}; 

\path[->,draw,thick] (b) to (a);
\path[->,draw,thick] (d) to (c);

\path[->,draw,thick] (a) to node[fill=white] {$P$} (c);
\path[->,draw,thick] (d) to node[fill=white] {$P$} (a);

\node[fill=white] at (1.5,.5)  {$P$}; 
\node[fill=white] at (2,0)  {$P$}; 

  \end{tikzpicture}
\end{minipage}
\end{center}
    \caption{a margin graph and the associated strict social preference relation for Split Cycle.}
    \label{SCexample}
\end{figure}

To help see why Split Cycle CCRs are AS representable, we note the following lemma that Holliday and Pacuit \citeyearpar[Lemma 3.16]{HP2020a} use to relate Split Cycle to the Beat Path voting method (\citealt{Schulze2011}).

\begin{lemma}\label{PathLem} If $f$ is a Split Cycle CCR, then for any profile $\mathbf{R}$ and $x,y\in X$, we have $xP(f(\mathbf{R}))y$ if and only if $\mathrm{Margin}_\mathbf{R}(x,y)>\mathrm{Margin}_\mathbf{R}(y,x)$ and
\[\mathrm{Margin}_\mathbf{R}(x,y)>Split\#(\rho)\mbox{ for every majority path }\rho\mbox{ from $y$ to $x$}.\]
\end{lemma}
Observe that the left-hand side of the displayed inequality depends only on the restricted profile $\mathbf{R}|_{\{x,y\}}$. Moreover, the right-hand side depends only on $\mathbf{R}^{-x,y}$: where $\rho=\langle z_1,\dots,z_n\rangle$ is a majority path in $\mathbf{R}$, we cannot have $\langle y,x\rangle=\langle z_i,z_{i+1}\rangle$ for any $i<n$, given $\mathrm{Margin}_\mathbf{R}(x,y)>\mathrm{Margin}_\mathbf{R}(y,x)$, and we cannot have $\langle x,y\rangle=\langle z_i,z_{i+1}\rangle$ because $\rho$ is a majority path from $y$ to $x$, which by definition cannot contain $x$ followed by $y$. Thus, we are only concerned with majority paths from $y$ to $x$ in $\mathbf{R}^{-x,y}$. 

The foregoing observation is the key to the following.

\begin{restatable}{proposition}{SplitAS}\label{SplitAS} Every Split Cycle CCR is AS representable with
\[\mathsf{Advantage}(x,y,\mathbf{R}|_{\{x,y\}}):=\mathrm{Margin}_{\mathbf{R}|_{\{x,y\}}}(x,y)\]
and
\[\mathsf{Standard}(x,y,\mathbf{R}^{-x,y})=\mbox{max}\{\mathrm{Split}\#_{\mathbf{R}^{-x,y}}(\rho)\mid \rho\mbox{ a majority path from $y$ to $x$ in $\mathbf{R}^{-x,y}$}\}.\]
\end{restatable}
\noindent Thus, according to Split Cycle viewed in terms of the AS model, the advantage of $x$ over $y$ is the margin, and we raise the standard for $x$ to defeat $y$ to that level $n$ at which there are no majority paths from $y$ to $x$ in which each candidate is majority preferred by more than $n$ to the next candidate in the cycle.

\begin{remark}\label{SCratio}
Defining Split Cycle as in Definition \ref{SplitCycleDef} but with $Ratio$ instead of $Margin$\footnote{Thus, a majority path, a majority cycle, and the splitting number of a majority path are assumed to be defined analogously with $Ratio$ instead of $Margin$.}  produces an AS representable CCR, where the underlying advantage and standard functions---taking values in $\mathbb{Q}_{>0}$ with its usual multiplication operation and ordering---are as in Proposition \ref{SplitAS} but with $Ratio$ instead of $Margin$.
\end{remark}

This completes our tour of three AS representable CCRs. These CCRs are not mere examples designed to prove the consistency of certain axioms; they have been independently proposed and studied for actual applications to voting (again see \citealt{Miller1980}, \citealt{Tideman1987}, \citealt{HP2020a,HP2020b}). While IIA may permit some CCRs that are usable in very special circumstances, such as the unanimity CCR, it is a point in favor of AS representability that it permits realistic general-purpose CCRs.

\section{AS Representability and Other Axioms}\label{explanatory}

Not only does AS representability make room for appealing CCRs, but also thinking in terms of AS representability helps explain certain properties of such CCRs. In this section, we give an example of this phenomenon. The properties in our example are traditionally stated not for CCRs but  for functions that assign to each profile a set of winning candidates. But given a CCR $f$, profile $\mathbf{R}$, and candidate $x\in X$, we can say that \textit{$x$ wins in $\mathbf{R}$ according to $f$} if $x$ is maximal in $P(f(\mathbf{R}))$, i.e., there is no $y\in X$ with $yP(f(\mathbf{R}))x$. Then we are interested in two properties concerning the choice of winners.

First, recall from Example \ref{DodgsonExample} that a candidate $x$ is a \textit{Condorcet winner} in a profile $\mathbf{R}$ if $\mathrm{Margin}_\mathbf{R}(x,y)>0$ for all $y\in X\setminus \{x\}$. 

\begin{definition} A CCR $f$ is \textit{Condorcet consistent} if for all profiles $\mathbf{R}$, if there is a Condorcet winner $x$ in $\mathbf{R}$, then $x$ is the unique winner in $\mathbf{R}$ according to $f$.\end{definition}
\noindent Note that CCRs that are AS representable using the $Margin$ function as \textsf{Advantage} will automatically satisfy at least the weaker property that if there is a Condorcet winner $x$, then $x$ is \textit{among} the winners in $\mathbf{R}$ according to $f$, since no $y$ will defeat $x$.

Second, the axiom of \textit{positive involvement} (see, e.g., \citealt{Saari1995}, \citealt{Perez2001}, \citealt{HP2021PI}) states that if $x$ wins in a profile $\mathbf{R}$, and $\mathbf{R}'$ is obtained from $\mathbf{R}$  by the addition of a voter who ranks $x$ strictly above all other candidates in $\mathbf{R}'$, then $x$ should still win in $\mathbf{R}'$. As Perez \citeyearpar[p.~605]{Perez2001} remarks, this seems to be ``the minimum to require concerning the coherence in the winning set when new voters are added.''  In the setting of this paper with a fixed electorate $V$, we cannot actually add a voter from $\mathbf{R}$ to $\mathbf{R}'$, but we can model this with a voter who casts a fully indifferent ballot in $\mathbf{R}$ changing her ballot in $\mathbf{R}'$ so that $x$ is ranked strictly above all other candidates.

\begin{definition}\label{PosInv}A CCR $f$ satisfies \textit{positive involvement} if for all profiles $\mathbf{R}$ and $\mathbf{R}'$ and $x\in X$, if $x$ wins in $\mathbf{R}$ according to $f$, and there is some $i\in V$ such that $\mathbf{R}_i= X\times X$, $xP(\mathbf{R}_i')y$ for all $y\in X\setminus\{x\}$, and $\mathbf{R}'_j=\mathbf{R}_j$ for all $j\in V\setminus\{i\}$, then $x$ wins in $\mathbf{R}'$ according to $f$.
\end{definition}

When positive involvement is violated, a voter ranking a candidate in first place causes that candidate to go from winning to losing, a highly counterintuitive result dubbed the Strong No Show Paradox (\citealt{Perez2001}). It is therefore striking that in \citealt{Perez2001}, only one Condorcet consistent voting method was known to satisfy positive involvement: the Minimax method (\citealt{Simpson1969}, \citealt{Kramer1977}). Minimax selects as winners in a profile $\mathbf{R}$ the candidates whose \textit{worst loss} is smallest, i.e., those $x$ who minimize the value of $\mathrm{max}\{\mathrm{Margin}_\mathbf{R}(y,x)\mid y\in X\}$, called the \textit{Minimax score} of $x$ (below we will redescribe Minimax as a CCR). More recently, Split Cycle was identified as a Condorcet consistent method satisfying positive involvement (\citealt{HP2020a}). 

Do Minimax and Split Cycle have something in common that explains their both satisfying positive involvement? It is not just that both are margin-based voting methods, since the same is true of Beat Path (\citealt{Schulze2011}) and Ranked Pairs, both of which violate positive involvement. The answer is rather that both Minimax and Split Cycle may be regarded as CCRs that are AS representable by the margin $\mathsf{Advantage}$ function and a \textit{continuous} $\mathsf{Standard}$ function, which implies positive involvement.

\begin{definition}\label{ContinuousDef} Suppose a CCR $f$ is AS representable using a standard function $\mathsf{Standard}$ as in Definition \ref{advstddef} where the relevant totally ordered group is $(\mathbb{Z},+,\leq)$. Then we say that $\mathsf{Standard}$ is \textit{continuous} if for any profiles $\mathbf{R}$ and $\mathbf{R}'$ and $i\in V$, if $\mathbf{R}_j=\mathbf{R}'_j$ for all $j\in V\setminus\{i\}$, then for any $x,y\in X$,
\[|\mathsf{Standard}(x,y,\mathbf{R}^{-x,y}) - \mathsf{Standard}(x,y,\mathbf{R}'^{-x,y}) |\leq 1.\]
That is, changing only one voter's ballot can change the standard by at most one.\end{definition}

\noindent Clearly the $\mathsf{Standard}$ in the AS representation of Split Cycle in Theorem \ref{SplitAS} is continuous: changing one voter can only change the weakest margin in a majority path by one.

\begin{proposition}\label{PosInvProp} Let $f$ be an AS representable CCR with $Margin$ as \textsf{Advantage} and a continuous $\mathsf{Standard}$ function as in Definition \ref{ContinuousDef}. Then $f$ satisfies positive involvement.\footnote{In fact, our proof shows that $f$ satisfies a stronger property that Ding et al.~\citeyearpar{Ding2022} call \textit{positive involvement in defeat}: if $y$ does not defeat $x$ in a profile $\mathbf{R}$, and $\mathbf{R}'$ is obtained from $\mathbf{R}$ by adding one new voter who ranks $x$ above $y$, then $y$ still does not defeat $x$ in $\mathbf{R}'$.}
\end{proposition}

\begin{proof} Suppose $\mathbf{R}$, $\mathbf{R}'$, $x$, and $i$ are as in Definition \ref{PosInv}. Further suppose that $x$ is not a winner in $\mathbf{R}'$, so there is some $y\in X$ such that $yP(f(\mathbf{R}'))x$. Then by our assumption about the AS representation of $f$, we have
$\mathrm{Margin}_{\mathbf{R}'}(y,x)> \mathsf{Standard}(y,x,\mathbf{R}'^{-x,y})$. 
Since $i$ is indifferent between $x$ and $y$ in $\mathbf{R}$ and ranks $x$ strictly above $y$ in $\mathbf{R}'$, we have $\mathrm{Margin}_\mathbf{R}(y,x)=\mathrm{Margin}_{\mathbf{R}'}(y,x)+1$. Then since $\mathsf{Standard}$ is continuous, it follows that $\mathrm{Margin}_\mathbf{R}(y,x)> \mathsf{Standard}(y,x,\mathbf{R}^{-x,y})$, so $yP(f(\mathbf{R}))x$. Hence $x$ is not a winner in $\mathbf{R}$.\end{proof}

Thus, by the observation before Proposition \ref{PosInvProp}, Split Cycle satisfies positive involvement. As for Minimax, we need to redescribe Minimax as an appropriate CCR.  There are various CCRs that yield the same winners as Minimax. For example,  there is the CCR that ranks candidates by Minimax scores, where lower scores are better. But note that this CCR may rank $x$ above $y$ even if a majority of voters prefer $y$ to $x$. Another CCR $f_{\mathrm{M}}$ that yields the same winners as Minimax but  ranks $x$ above $y$ only if a majority of voters prefer $x$ to $y$ can be defined as follows: $xf_{\mathrm{M}}(\mathbf{R})y$ if $x=y$ or $x P(\mathbf{R}) y$ where  
\[xP(\mathbf{R})y\mbox{ if } \mathrm{Margin}_\mathbf{R}(x,y)> \mathrm{min}\big(\big\{ \mathrm{max}\big(\{ \mathrm{Margin}_{\mathbf{R}}(w,z)\mid w\in X \}\big) \mid z\in X \big\}\big).\]
In other words, if the margin of $x$ over $y$ is greater than the smallest Minimax score of any candidate, then $x$ defeats $y$.\footnote{Note that the smallest Minimax score of any candidate is non-negative because for any $z$, $\mathrm{Margin}_{\mathbf{R}}(z,z)=0$  so $\mathrm{max}\big(\{ \mathrm{Margin}_{\mathbf{R}}(w,z)\mid w\in X \})\geq 0$.} Note that we can equivalently replace $\mathbf{R}$ on the right of $>$ with $\mathbf{R}^{-x,y}$, so with the $\mathsf{Standard}$ determined in this way, $f_{\mathrm{M}}$ is AS representable with  $Margin$ as $\mathsf{Advantage}$.\footnote{Also note that while $f_{\mathrm{M}}$ is not acyclic, there is always a nonempty set of winners according to $f_{\mathrm{M}}$, namely the set of Minimax winners.} Moreover, the smallest Minimax score of a candidate can change by at most one as a result of changing only one voter's ballot. Thus, the $\mathsf{Standard}$ function for $f_{\mathrm{M}}$ is continuous. Hence  $f_{\mathrm{M}}$ satisfies positive involvement by Proposition \ref{PosInvProp}. Then since $f_{\mathrm{M}}$ selects the same winners as Minimax,  Minimax satisfies positive involvement.

By contrast, the $\mathsf{Standard}$ functions used in the AS representations of the Gillies Covering and Ranked Pairs CCRs in Propositions \ref{GilliesAS} and \ref{RankedPairsAS} are not continuous. Indeed, since these CCRs violate positive involvement (\citealt{Perez2001}, \citealt{HP2021PI}), by Proposition~\ref{PosInvProp} there are no AS representations of these CCRs with $Margin$ as $\mathsf{Advantage}$ and with a continuous $\mathsf{Standard}$ function. 

Proposition \ref{PosInvProp} not only explains what Split Cycle and Minimax have in common in virtue of which they satisfy positive involvement but also provides guidance in the search for new CCRs that satisfy the axiom. That being said, there are other reasons an AS representable CCRs may satisfy positive involvement. Ding et al.~\citeyearpar{Ding2022} observe that unlike Gillies Covering, the Weighted Covering CCR (\citealt{Dutta1999}, \citealt{Fernandez2018})  satisfies positive involvement. Moreover, Weighted Covering can be AS represented in the same way as Gillies Covering in Proposition \ref{GilliesAS} only using the weighted covering relation instead of the covering relation, where $x$ covers $y$ in the weighted sense if for all $z\in X$, $\mathrm{Margin}_\mathbf{R}(x,z)\geq \mathrm{Margin}_\mathbf{R}(y,z)$. This AS representation uses $Margin$ as $\mathsf{Advantage}$ and a non-continuous $\mathsf{Standard}$; moreover,  it is unclear whether there is any AS representation of Weighted Covering using $Margin$ as $\mathsf{Advantage}$ and a continuous $\mathsf{Standard}$.  

Adopting the AS representation of Weighted Covering analogous to Proposition \ref{GilliesAS}, the reason Weighted Covering satisfies positive involvement is that when a voter switches from a fully indifferent ballot in $\mathbf{R}$ to a ballot with $x$ in first place in $\mathbf{R}'$, $\mathsf{Standard}(y,x,\mathbf{R}'^{-y,x})$ cannot decrease at all from $\mathsf{Standard}(y,x,\mathbf{R}^{-y,x})$ (unlike when the $\mathsf{Standard}$ is set with the unweighted covering relation), as shown in \citealt[Prop.~3.16]{Ding2022}. Then since the margin of $y$ over $x$ decreases by $1$, if $y$ did not defeat $x$  in $\mathbf{R}$, then $y$  still does not defeat $x$ in $\mathbf{R}'$, so positive involvement holds. Thus, despite being non-continuous in general, the  $\mathsf{Standard}$ for Weighted Covering has the right behavior in the specific case of a voter switching from indifference to ranking $x$ on top. By contrast, for Split Cycle and Minimax, $\mathsf{Standard}(y,x,\mathbf{R}'^{-y,x})$ \textit{can} decrease from $\mathsf{Standard}(y,x,\mathbf{R}^{-y,x})$ when a voter switches from indifference in $\mathbf{R}$ to ranking $x$ on top in $\mathbf{R}'$, but only by $1$ thanks to continuity, and then since the margin of $y$ over $x$ decreases by $1$ as well,  $y$ still does not defeat $x$. Thus, we see that Split Cycle and Minimax satisfy positive involvement for a different reason than Weighted Covering does.

Stepping back from the specific case of positive involvement, this discussion suggests that we should not stop at studying how AS representability relates to other axioms, which CCRs are AS rationzaliable, etc. We can gain additional insight by studying how the properties of a CCR relate to the properties of  $\mathsf{Advantage}$ and $\mathsf{Standard}$ functions that represent it. Such correspondences also open up a new form of argument for axioms on CCRs: argue that the $\mathsf{Advantage}$ and $\mathsf{Standard}$ functions ought to satisfy certain properties and then show that being AS representable with such functions  entails satisfying the relevant axiom.

\section{Conclusion}\label{Conclusion}

In one sense, there is no escape from Arrow's Impossibility Theorem, just as there is no escape from any other mathematical result. But in another sense, we have argued that there is an escape: there is a well-motivated way to weaken Arrow's assumptions, without drastically violating the idea behind Arrow's IIA, that opens up the possibility of appealing CCRs without dictators, vetoers, etc. The insight and much of the intuition behind IIA can be captured with the weaker axiom of advantage-standard representability: while the intrinsic advantage of one candidate over another should depend only on the voters' preferences between those two candidates, the standard required for strict social preference might be context dependent. Furthermore, Arrow's argument for the completeness of social preference failed to consider maximal element choice, which allows for social choice from any choice set even without a complete underlying ranking. Having motivated our relaxation of Arrow's IIA and social rationality assumptions, we provided three examples of CCRs that satisfy advantage-standard representability, anonymity, neutrality, Pareto, and have no vetoers---two of which are acyclic and one of which is transitive. Thus, we conclude that there is reason to be optimistic in the face of Arrow's Impossibility Theorem. While the theorem shows that we cannot accept the letter of Arrow's axioms, we believe that the Advantage-Standard model shows that we can accept much of their spirit.

\subsection*{Acknowledgements}

For helpful feedback, we thank Yifeng Ding, Eric Pacuit, the two anonymous referees, and the audiences at the TEAM 2020 conference at Princeton University and the 2022 Meeting of the Society for Social Choice and Welfare, where this work was presented.

\appendix

\section{Appendix}

\subsection{Totally ordered groups}\label{GroupAppendix}

We recall that a \textit{group} is a pair $(G,\circ)$ where $G$ is a nonempty set and $\circ$ is an associative binary operation on $G$ for which there is a unique $e\in G$ such that $e\circ a=a\circ e=a$ for all $a\in G$, and for each $a\in G$, there is a unique $a^{-1}\in G$ such that $a\circ a^{-1}=a^{-1}\circ a = e$. The element $e$ is the \textit{identity element} of the group,  and $a^{-1}$ is the \textit{inverse of} $a$. A \textit{totally ordered group} is a triple $(G,\circ,\leq)$ where $(G,\circ)$ is a group, $\leq$ is a binary relation on $G$ that is transitive, complete, and antisymmetric (if $a\leq b$ and $b\leq a$, then $a=b$), and for all $a,b,c\in G$, if $a\leq b$, then $c\circ a\leq c\circ b$ and $a\circ c\leq b\circ c$. Let $a<b$ mean that $a\leq b$ and $b\not \leq a$. We use a basic fact about the relation between the identity, inverse, and order relation.

\begin{lemma}\label{GroupLem} For any totally ordered group $(G,\circ,\leq)$ with identity element $e$ and $a\in G$, if $e< a$ then $a^{-1}<e$.
\end{lemma}

\begin{proof} We show that if $a^{-1}\not <e$, then $e\not < a$. By completeness of $\leq$, this is equivalent to: if $e\leq a^{-1}$, then $a\leq e$. Indeed, from $e\leq a^{-1}$ we have $a\circ e\leq a\circ a^{-1}$ and hence $a\leq e$.\end{proof}

\subsection{Weak IIA and orderability}

\converse*
\begin{proof} Define a CCR $f$ as follows:
\begin{itemize}
      \item[(a)] if there is an $i\in V$ and some enumeration $z_1,\dots,z_n$ of $X$  such that $\mathbf{R}$ is of the form shown on the left below, then the social relation is as shown on the right below:
          \begin{center}
    \begin{tabular}{c|c}
    $V\setminus\{i\}$ & $\{i\}$ \\
    \hline
             $z_{n-1}$ & $z_n$\\
         $z_n$ & $z_{n-1}$\\
        $z_1$ & $z_1$  \\
         \vdots & \vdots \\
         $z_{n-2}$ & $z_{n-2}$ 
    \end{tabular}\qquad \begin{tabular}{c}
    $f(\mathbf{R})$\\
    \hline
             $z_{n-1}$\\
         $z_n$\\
        $z_1$  \\
         \vdots \\
         $z_{n-2}$ \\
    \end{tabular}
    \end{center}
      \item[(b)] if there is a partition of $V$ into $C_1$ and $C_2$ with $|C_1|>|C_2|>1$ and some enumeration $z_1,\dots,z_n$ of $X$ such that $\mathbf{R}$ is of the form shown on the left below, then the social relation is as shown on the right below:
          \begin{center}
    \begin{tabular}{c|c}
    $C_1$ & $C_2$ \\
    \hline
        $z_1$ & $z_1$  \\
         \vdots & \vdots \\
         $z_{n-2}$ & $z_{n-2}$ \\
         $z_{n-1}$ & $z_n$\\
         $z_n$ & $z_{n-1}$\\
    \end{tabular}\qquad \begin{tabular}{c}
    $f(\mathbf{R})$\\
    \hline
        $z_1$  \\
         \vdots \\
         $z_{n-2}$ \\
         $z_{n-1}$\\
         $z_n$\\
    \end{tabular}
    \end{center}

    \item[(c)] otherwise $x f(\mathbf{R}) y$ if and only if $x\mathbf{R}_iy$ for all $i\in V$.
\end{itemize}

By its definition, $f$ is anonymous and neutral. Since $f(\mathbf{R})$ is a linear order in cases (a) and (b), and since the Pareto CCR used in case (c) is transitive, $f$ is transitive. Obviously $f$ also satisfies Pareto. Next observe that $xP(f(\mathbf{R}))y$ only if a majority of voters rank $x$ over $y$. It follows that $f$ satisfies weak IIA.\footnote{Note that  $xI(f(\mathbf{R}))y$ only if all voters are indifferent between $x$ and $y$. Thus $f$ satisfies not only weak IIA but PN-weak IIA (recall Definition \ref{PNPIdef}).} However, we claim that $f$ is not orderable. Suppose toward a contradiction that $f$ is orderable, so that for each $x,y$, there is a transitive and complete relation $\leqslant_{x,y}$ on $\mathcal{P}^+(x,y)$ satisfying the condition in Definition \ref{orderable}. Partition $V$ into $S_1$ and $S_2$ such that  $|S_2|=2$. Let $\mathbf{R}^1$ and $\mathbf{R}^2$ be profiles of the following forms for some $j\in S_2$:
\begin{center}
 \begin{tabular}{c|c}
    $S_1$ & $S_2$ \\
    \hline
    $x$ & $y$ \\
    $y$ & $x$ \\
        $z_1$ & $z_1$  \\
         \vdots & \vdots \\
         $z_{n-2}$ & $z_{n-2}$ \\
    \end{tabular}\qquad    \begin{tabular}{c|c}
    $S_1\cup \{j\}$ & $S_2\setminus\{j\}$ \\
    \hline
    $x$ & $y$ \\
    $y$ & $x$ \\
        $z_1$ & $z_1$  \\
         \vdots & \vdots \\
         $z_{n-2}$ & $z_{n-2}$ \\
    \end{tabular}
    \end{center}
    Then $x N(f(\mathbf{R}^1))y$ by (c) but $x P(f(\mathbf{R}^2))y$ by (a).  Now let $\mathbf{R}^3$ and $\mathbf{R}^4$ be the following profiles:
    \begin{center}
         \begin{tabular}{c|c}
    $S_1$ & $S_2$ \\
    \hline
        $z_1$ & $z_1$  \\
         \vdots & \vdots \\
         $z_{n-2}$ & $z_{n-2}$ \\
         $x$ & $y$\\
         $y$ & $x$\\
    \end{tabular}\qquad         \begin{tabular}{c|c}
    $S_1\cup \{j\}$ & $S_2\setminus\{j\}$ \\
    \hline
        $z_1$ & $z_1$  \\
         \vdots & \vdots \\
         $z_{n-2}$ & $z_{n-2}$ \\
         $x$ & $y$\\
         $y$ & $x$\\
    \end{tabular}
    \end{center}
    Then $xP(f(\mathbf{R}^3))y$ by (b) but $xN(f(\mathbf{R}^4))y$ by (c).
    
    Since $\mathbf{R}^{1{-x,y}}=\mathbf{R}^{2{-x,y}}$ and $xP(f(\mathbf{R}^2)y$ while $xN(f(\mathbf{R}^1))y$, orderability implies that
    \[\textit{ not }\mathbf{R}^2|_{\{x,y\}}\leqslant_{x,y}\mathbf{R}^1|_{\{x,y\}}.\]
    Similarly, since $\mathbf{R}^{3{-x,y}}=\mathbf{R}^{4{-x,y}}$ and $xP(f(\mathbf{R}^3)y$ while $xN(f(\mathbf{R}^4))y$, orderability implies that 
    \[\textit{ not } \mathbf{R}^3|_{\{x,y\}}\leqslant_{x,y}\mathbf{R}^4|_{\{x,y\}}.\]
   Since $\mathbf{R}^{3}|_{\{x,y\}}=\mathbf{R}^1|_{\{x,y\}}$ and $\mathbf{R}^{4}|_{\{x,y\}}=\mathbf{R}^2|_{\{x,y\}}$, it follows that
        \[\textit{ not } \mathbf{R}^1|_{\{x,y\}}\leqslant_{x,y}\mathbf{R}^2|_{\{x,y\}}\]
   which contradicts the completeness of $\leqslant_{x,y}$. 
   \end{proof}

\subsection{Characterization of AS representability}

\begin{lemma}\label{PNweakIIAtoAdvStd}
If $f$ is orderable and satisfies weak IIA, then $f$ is AS representable. 
\end{lemma}
\begin{proof}

We build integer-valued $\mathsf{Advantage}$ and $\mathsf{Standard}$ functions satisfying $(\ref{minus})$ and $(\ref{minimal})$ that represent $f$ in the sense of (\ref{Piff}). For each $\left\langle x,y\right\rangle\in X^2$, there is a transitive and complete order $\leqslant_{x,y}$ on $\mathcal{P}^+(x,y)$ by orderability. For $\mathbf{Q},\mathbf{Q}'\in \mathcal{P}^+(x,y)$, define $\mathbf{Q}\sim \mathbf{Q}'$ if and only if $\mathbf{Q}\leqslant_{x,y}\mathbf{Q}'$ and $\mathbf{Q}'\leqslant_{x,y}\mathbf{Q}$. Then $\sim$ is an equivalence relation, and its equivalence classes---denote them $\mathcal{Q}_1,\ldots, \mathcal{Q}_n$---can be ordered such that $\mathbf{Q}\in \mathcal{Q}_i$ and $\mathbf{Q}'\in \mathcal{Q}_j$ with $i<j$ if and only if  $\mathbf{Q}\leqslant_{x,y}\mathbf{Q}'$ but \textit{not} $\mathbf{Q}'\leqslant_{x,y}\mathbf{Q}$. For each $i\in \{1,\ldots,n\}$ and $\mathbf{Q}\in \mathcal{Q}_i$, define $\mathsf{Advantage}(x,y,\mathbf{Q})=i$ and  $\mathsf{Advantage}(y,x,\mathbf{Q})=-i$. For every $\left\langle x,y\right\rangle\in X^2$ and $\{x,y\}$-profile $\mathbf{Q}$ not in $\mathcal{P}^+(x,y)\cup \mathcal{P}^+(y,x)$, set $\mathsf{Advantage}(x,y,\mathbf{Q})=\mathsf{Advantage}(y,x,\mathbf{Q})=0$. Note that $\mathcal{P}^+(x,y)\cap \mathcal{P}^+(y,x)=\varnothing$ by weak IIA, so the advantages assigned for $\left\langle x,y\right\rangle$ and $\left\langle y,x\right\rangle$ do not conflict. 

Next, for each profile $\mathbf{R}$ and $\langle x,y\rangle\in X^2$, we define $\mathsf{Standard}(x,y,\mathbf{R}^{-x,y})$. Find the least $j$, if there is one, such that for some profile $\mathbf{R}'$ with $\mathbf{R}'^{-x,y}=\mathbf{R}^{-x,y}$, we have ${\mathbf{R}'| _{x,y}\in\mathcal{Q}_j}$ and $xP(f(\mathbf{R}'))y$, and set $\mathsf{Standard}(x,y,\mathbf{R}^{-x,y})=j-1$. If there is no such $j$, set $\mathsf{Standard}(x,y,\mathbf{R}^{-x,y})=n$.

Note that the $\mathsf{Standard}$ function is non-negative and the $\mathsf{Advantage}$ function satisfies \[\mathsf{Advantage}(x,y,\mathbf{R}|_{\{x,y\}})=-\mathsf{Advantage}(y,x,\mathbf{R}|_{\{x,y\}}),\]
so $(\ref{minus})$ and $(\ref{minimal})$ are satisfied. 

Fix $\left\langle x,y\right\rangle\in X^2$ and a profile $\mathbf{R}$. Let $\mathcal{Q}_1,\ldots,\mathcal{Q}_n$ (resp.~$\mathcal{S}_1,\ldots,\mathcal{S}_m$) be the ordering of $\sim_{x,y}$ (resp.~$\sim_{y,x})$ equivalence classes of $\mathcal{P}^+(x,y)$ (resp.~$\mathcal{P}^+(y,x)$)  as discussed above. If $xP(f(\mathbf{R}))y$, then there is some $i\in \{1,\ldots,n\}$ such that $\mathbf{R}|_{\{x,y\}}\in \mathcal{Q}_i$, so by construction, $\mathsf{Advantage}(x,y,\mathbf{R}|_{\{x,y\}})=i$. Let $k$ be the least $j$ such that there is some profile $\mathbf{R}'$ with $\mathbf{R}'^{-x,y}=\mathbf{R}^{-x,y}$, ${\mathbf{R}'| _{x,y}\in\mathcal{Q}_j}$, and $xP(f(\mathbf{R}'))y$. Then since $\mathbf{R}|_{\{x,y\}}\in \mathcal{Q}_i$ and $xP(f(\mathbf{R}))y$, $k$ is at most $i$. Hence, by construction, $\mathsf{Standard}(x,y,\mathbf{R}^{-x,y})\leq i-1$. Thus,
 \[xP(f(\mathbf{R}))y \implies \mathsf{Advantage}(x,y,\mathbf{R}|_{\{x,y\}})>\mathsf{Standard}(x,y,\mathbf{R}^{-x,y}).\]
  
  Now assume \textit{not} $xP(f(\mathbf{R}))y$. If $yP(f(\mathbf{R}))x$, then there is some $i\in \{1,\ldots, m\}$ such that $\mathbf{R}|_{\{x,y\}}\in \mathcal{S}_i$, so $\mathsf{Advantage}(y,x,\mathbf{R}|_{\{x,y\}})=i$ and $\mathsf{Advantage}(x,y,\mathbf{R}|_{\{x,y\}})=-i<0$. Thus, 
   \[\mathsf{Advantage}(x,y,\mathbf{R}|_{\{x,y\}})\not> \mathsf{Standard}(x,y,\mathbf{R}^{-x,y})\geq 0.\]
If $xI(f(\mathbf{R}))y$ or $xN(f(\mathbf{R}))y$, then one of the following holds: 
    \begin{enumerate}
 \item  $\mathbf{R}|_{\{x,y\}}\notin \mathcal{P}^+(x,y)\cup \mathcal{P}^+(y,x)$,
 \item $\mathbf{R}|_{\{x,y\}}\in  \mathcal{P}^+(x,y)$, or
 \item  $\mathbf{R}|_{\{x,y\}}\in  \mathcal{P}^+(y,x)$.
 \end{enumerate}
  
  In the first case, $\mathsf{Advantage}(x,y,\mathbf{R}|_{\{x,y\}})=0$ by construction, so
    \[\mathsf{Advantage}(x,y,\mathbf{R}|_{\{x,y\}})\not> \mathsf{Standard}(x,y,\mathbf{R}^{-x,y})\geq 0.\]
  
  In the second case, there is some $i\in \{1,\ldots,n\}$  such that $\mathbf{R}|_{\{x,y\}}\in \mathcal{Q}_i$, and so $\mathsf{Advantage}(x,y,\mathbf{R}|_{\{x,y\}})=i$. We claim that \[\mathsf{Standard}(x,y,\mathbf{R}^{-x,y})\geq i.\]
    Indeed, if $\mathsf{Standard}(x,y,\mathbf{R}^{-x,y})=j\leq i-1$, then by construction there is a profile $\mathbf{R}'$ such that $\mathbf{R}'|_{\{x,y\}}\in\mathcal{Q}_{j+1}$, $\mathbf{R}'^{-x,y}=\mathbf{R}^{-x,y}$, and $xP(f(\mathbf{R}'))y$. Since $j+1\leq i$, we have $\mathbf{R}'|_{\{x,y\}}\leqslant_{x,y}\mathbf{R}|_{\{x,y\}}$, and so 
    $xP(f(\mathbf{R}'))y$ implies $xP(f(\mathbf{R}))y$. This contradicts our initial assumption that $xI(f(\mathbf{R}))y$ or $xN(f(\mathbf{R}))y$. Thus,
      \[i=\mathsf{Advantage}(x,y,\mathbf{R}|_{\{x,y\}})\not> \mathsf{Standard}(x,y,\mathbf{R}^{-x,y})\geq i.\]

  In the third case,  $\mathbf{R}|_{\{x,y\}}\in \mathcal{S}_i$ for some $i\in \{1,\ldots,m\}$, so  $\mathsf{Advantage}(y,x,\mathbf{R}|_{\{x,y\}})=i$ and $\mathsf{Advantage}(x,y,\mathbf{R}|_{\{x,y\}})=-i$. Thus,
     \[-i=\mathsf{Advantage}(x,y,\mathbf{R}|_{\{x,y\}})\not> \mathsf{Standard}(x,y,\mathbf{R}^{-x,y})\geq 0.\]
    We conclude that for any $\left \langle x,y\right\rangle\in X^2$ and profile $\mathbf{R}$, 
    \[xP(f(\mathbf{R}))y \iff \mathsf{Advantage}(x,y,\mathbf{R}|_{\{x,y\}})> \mathsf{Standard}(x,y,\mathbf{R}^{-x,y}),\]
    and hence  $f$ is AS representable.
    \end{proof}

\begin{lemma}\label{AdvStdtoorderable} If $f$ is AS representable, then $f$ is orderable.
\end{lemma}
\begin{proof}
Let $f$ be AS representable with \textsf{Advantage} and \textsf{Standard} functions taking values in $(G,\circ,\leq)$. Fix $x,y\in X$. For $\mathbf{Q}\in \mathcal{P}^+(x,y)$, let $\varphi(\mathbf{Q})=\mathsf{Advantage}(x,y,\mathbf{Q})$. If $\varphi(\mathbf{Q})\leq \varphi(\mathbf{Q}')$, this implies that for any profiles $\mathbf{R}^1,\mathbf{R}^2$ with $\mathbf{R}^1|_{\{x,y\}}=\mathbf{Q}$, $\mathbf{R}^2|_{\{x,y\}}=\mathbf{Q}'$, and $\mathbf{R}^{1-x,y}=\mathbf{R}^{2-x,y}$, if $\mathsf{Advantage}(x,y,\mathbf{Q})>\mathsf{Standard}(x,y,\mathbf{R}^{1-x,y})$, then also $\mathsf{Advantage}(x,y,\mathbf{Q}')>\mathsf{Standard}(x,y,\mathbf{R}^{2-x,y})$, since \[\mathsf{Standard}(x,y,\mathbf{R}^{1-x,y})=\mathsf{Standard}(x,y,\mathbf{R}^{2-x,y}).\]
As $f$ is AS representable, it follows that 
\begin{align}
xP(f(\mathbf{R}^1))y&\implies \mathsf{Advantage}(x,y,\mathbf{Q})>\mathsf{Standard}(x,y,\mathbf{R}^{1-x,y})\nonumber\\
&\implies\mathsf{Advantage}(x,y,\mathbf{Q}')>\mathsf{Standard}(x,y,\mathbf{R}^{2-x,y})\nonumber \\
&\implies xP(f(\mathbf{R}^2))y\label{PtoP}.
\end{align}
Define a transitive and complete ordering $\leqslant_{x,y}$ on $\mathcal{P}^+(x,y)$ by $\mathbf{Q}\leqslant_{x,y}\mathbf{Q}'$ if and only if $\varphi(\mathbf{Q})\leq \varphi(\mathbf{Q}')$. By (\ref{PtoP}), we have given an ordering of the kind required for $f$ to be orderable. Since $x,y\in X$ were arbitrary, we are done.
\end{proof}
Putting together Proposition \ref{PIweakIIAlemma} and Lemmas \ref{PNweakIIAtoAdvStd} and \ref{AdvStdtoorderable}, we obtain our characterization of AS representability.

\characterization*

\subsection{AS representable CCRs}

\GilliesTrans*

\begin{proof} Transitivity of $f_{cov}(\mathbf{R})$ is equivalent to the conjunction of transitivity of $P_{cov}(\mathbf{R})$, transitivity of $I_{cov}(\mathbf{R})$, and the  following:
\begin{itemize}
    \item IP-transitivity: $xI_{cov}(\mathbf{R})y$ and $yP_{cov}(\mathbf{R})z$, then $xP_{cov}(\mathbf{R})z$;
    \item PI-transitivity: $xP_{cov}(\mathbf{R})y$ and $yI_{cov}(\mathbf{R})z$, then $xP_{cov}(\mathbf{R})z$.
\end{itemize}
Transitivity of $P_{cov}(\mathbf{R})$ is well known and easy to check. Transitivity of $I_{cov}(\mathbf{R})$ is also clear. For IP-transitivity, suppose $xI_{cov}(\mathbf{R})y$ and $yP_{cov}(\mathbf{R})z$. Since  $yP_{cov}(\mathbf{R})z$, we have ${y\succ_\mathbf{R}z}$, which with $xI_{cov}(\mathbf{R})y$ implies $x\succ_\mathbf{R}z$. Now suppose $v\succ_\mathbf{R}x$. Then since $xI_{cov}(\mathbf{R})y$, we have $v\succ_\mathbf{R}y$. Then since $yP_{cov}(\mathbf{R})z$, we have $v\succ_\mathbf{R}z$. Thus, $xP_{cov}(\mathbf{R})z$. For {PI-transitivity}, the argument is similar.\end{proof}

\GilliesAS*
\begin{proof}
$\mathsf{Advantage}$ and $\mathsf{Standard}$ take values in $\mathbb{Z}$ with the usual addition operation and ordering. It is clear that (\ref{minus}) and (\ref{minimal}) of Definition \ref{advstddef} hold. For $(\ref{Piff})$,  assume $xP(f_{cov}(\mathbf{R}))y$. Then $\mathsf{Advantage}(x,y,\mathbf{R}|_{\{x,y\}})>0$ and  $\mathsf{Standard}(x,y,\mathbf{R}^{-x,y})=0$. Assume \textit{not} $xP(f_{cov}(\mathbf{R}))y$. Then either $x\not\succ_{\mathbf{R}}y$ or $x\succ_{\mathbf{R}}y$ but there is some ${v\in X\setminus \{x,y\}}$ such that $v\succ_{\mathbf{R}}x$ and $v\not\succ_{\mathbf{R}}y$. In the first case, we have $\mathsf{Advantage}(x,y,\mathbf{R}|_{\{x,y\}})\leq 0$ and hence $\mathsf{Advantage}(x,y,\mathbf{R}|_{\{x,y\}})\leq \mathsf{Standard}(x,y,\mathbf{R}^{-x,y})$. In the second case, we have $\mathsf{Standard}(x,y,\mathbf{R}^{-x,y})= |V|$ and hence $\mathsf{Advantage}(x,y,\mathbf{R}|_{\{x,y\}})\leq \mathsf{Standard}(x,y,\mathbf{R}^{-x,y})$. Thus, $(\ref{Piff})$ holds.
\end{proof}

\GilliesPN*

\begin{proof}
Assume $\mathbf{R}|_{\{x,y\}}=\mathbf{R}'|_{\{x,y\}}$ and $xP(f_{cov}(\mathbf{R}))y$. Then $x\succ_{\mathbf{R}}y$ and for all $v\in X$, $v\succ_{\mathbf{R}}x$ implies $v\succ_{\mathbf{R}}y$. Note that \textit{not} $yP(f_{cov}(\mathbf{R}'))x$ since  $x\succ_{\mathbf{R}'}y$. Furthermore, \textit{not} $xI(f_{cov}(\mathbf{R}'))y$ since $x\succ_{\mathbf{R}'}y$ but it is not the case that $x\succ_{\mathbf{R}}x$. Thus, either $xP(f_{cov}(\mathbf{R}'))y$ or $xN(f_{cov}(\mathbf{R}'))y$. 
\end{proof}

\RankedPairsAS*

\begin{proof}
Let $f$ be a Ranked Pairs CCR. $\mathsf{Advantage}$ and $\mathsf{Standard}$ take values in $\mathbb{Z}$ with the usual addition operation and ordering. Clearly, condition $(\ref{minus})$ of Definition \ref{advstddef} holds. As for (\ref{minimal}), let $\mathcal{M}^{x,y}_k(\mathbf{R})$ denote $\mathcal{M}(\mathbf{R}^{-x,y})+x\overset{k}{\to}y$ and note that for any $T\in \mathcal{L}(X^2\setminus \Delta_X)$, 
$\mathbb{RP}(\mathcal{M}^{x,y}_k(\mathbf{R}),T)$ only contains edges with positive weights by definition. Thus,  $\mathsf{Standard}(x,y,\mathbf{R}^{-x,y})\geq 0$ for all $(x,y,\mathbf{R}^{-x,y})\in D_S$,  so (\ref{minimal}) holds. For (\ref{Piff}), assume $xP(f(\mathbf{R}))y$ so that ${\left\langle x,y\right\rangle \in \mathbb{RP}(\mathcal{M}(\mathbf{R}))}$. Then \[\mathsf{Advantage}(x,y,\mathbf{R}|_{\{x,y\}})-1\in \{k-1 \mid \langle x,y\rangle\in \mathbb{RP}(\mathcal{M}^{x,y}_k(\mathbf{R}))\},\]
since where $k=\mathsf{Advantage}(x,y,\mathbf{R}|_{\{x,y\}})$, we have $\mathcal{M}(\mathbf{R})=\mathcal{M}^{x,y}_k(\mathbf{R})$. Thus, we conclude that $\mathsf{Advantage}(x,y,\mathbf{R}|_{\{x,y\}})>\mathsf{Standard}(x,y,\mathbf{R}^{-x,y})$.

Assume \textit{not} $xP(f(\mathbf{R}))y$ and so  $\left\langle x,y\right\rangle\notin \mathbb{RP}(\mathcal{M}(\mathbf{R}))$. First note that 
\[\mathsf{Advantage}(x,y,\mathbf{R}|_{\{x,y\}})-1\notin \{k-1 \mid \langle x,y\rangle\in \mathbb{RP}(\mathcal{M}_k^{x,y}(\mathbf{R}))\},\]
since where $k=\mathsf{Advantage}(x,y,\mathbf{R}|_{\{x,y\}})$, we have $\mathcal{M}(\mathbf{R})=\mathcal{M}^{x,y}_k(\mathbf{R})$. We claim it follows that $j\notin \{k-1 \mid \langle x,y\rangle\in \mathbb{RP}(\mathcal{M}^{x,y}_k(\mathbf{R}))\}$ for all $j<\mathsf{Advantage}(x,y,\mathbf{R}|_{\{x,y\}})-1$, which would complete the proof. It suffices to show that for any $j<\mathsf{Advantage}(x,y,\mathbf{R}|_{\{x,y\}})-1$, if $\left\langle x,y\right\rangle \in \mathbb{RP}(\mathcal{M}^{x,y}_j(\mathbf{R}))$, then $\left\langle x,y\right\rangle \in\mathbb{RP}(\mathcal{M}^{x,y}_{j+1}(\mathbf{R}))$. To see this, fix $T\in \mathcal{L}(X^2\setminus \Delta_X)$ and note that since $\left\langle x,y\right\rangle \in \mathbb{RP}(\mathcal{M}^{x,y}_j(\mathbf{R}),T)$, there is some $n$ such $\left\langle x,y\right\rangle \in \mathbb{RP}(\mathcal{M}^{x,y}_j(\mathbf{R}),T)_{n+1}\setminus \mathbb{RP}(\mathcal{M}^{x,y}_j(\mathbf{R}),T)_n$, which means that  $\mathbb{RP}(\mathcal{M}^{x,y}_j(\mathbf{R}),T)_n\cup \{\left\langle x,y\right\rangle\}$ is acyclic. Considering now when the margin between $x$ and $y$ is $j+1$, $\left\langle x,y\right\rangle$ will be the maximum element of $E_{\mathcal{M}^{x,y}_{j+1}(\mathbf{R})}\setminus C_m$ according to $>_{\mathcal{M}^{x,y}_{j+1,T}}$ for some $m\leq n$. Since 
\[\mathbb{RP}(\mathcal{M}^{x,y}_{j+1}(\mathbf{R}),T)_m=\mathbb{RP}(\mathcal{M}^{x,y}_{j}(\mathbf{R}),T)_m\subseteq \mathbb{RP}(\mathcal{M}^{x,y}_{j}(\mathbf{R}),T)_n,\]
it follows that $\mathbb{RP}(\mathcal{M}^{x,y}_{j+1}(\mathbf{R}),T)_m\cup \{\left\langle x,y\right\rangle\}$ is acyclic as well. Note that the equality in the displayed line above follows from the fact that all the same edges are considered up to stage $m$ when the Ranked Pairs algorithm is applied to  $\mathcal{M}^{x,y}_{j+1}(\mathbf{R})$ and $\mathcal{M}^{x,y}_{j}(\mathbf{R})$ with $T$. Thus,
\[\left\langle x,y\right\rangle\in \mathbb{RP}(\mathcal{M}^{x,y}_{j+1}(\mathbf{R}),T)_{m+1}\subseteq \mathbb{RP}(\mathcal{M}^{x,y}_{j+1}(\mathbf{R}),T).\]
Since $T$ was arbitrary, we are done.
\end{proof}
\SplitAS*
\begin{proof}
Let $f$ be a Split Cycle CCR. $\mathsf{Advantage}$ and $\mathsf{Standard}$ take values in $\mathbb{Z}$ with the usual addition operation and ordering. Clearly, $(\ref{minus})$ of Definition \ref{advstddef} holds. As for (\ref{minimal}), note that $\mathrm{Margin}_{\mathbf{R}^{-x,y}}(z_i,z_{i+1})\geq 0$ for any $z_i,z_{i+1}$ in a majority path from $y$ to $x$ in $\mathbf{R}^{-x,y}$ since $z_i\succ_\mathbf{R}z_{i+1}$. It follows that $\mathrm{Split}\#_{\mathbf{R}^{-x,y}}(\rho)\geq 0$ for any majority path $\rho$ from $y$ to $x$ in $\mathbf{R}^{-x,y}$, and so $\mathsf{Standard}(x,y,\mathbf{R}^{-x,y})\geq 0$ for any $(x,y,\mathbf{R}^{-x,y})\in D_S$.

For (\ref{Piff}), if $xP(f(\mathbf{R}))y$, then $\mathrm{Margin}_\mathbf{R}(x,y)>\mathrm{Margin}_\mathbf{R}(y,x)$ and
\begin{equation}\mathrm{Margin}_\mathbf{R}(x,y)> \mbox{max}\{\mathrm{Split}\#_\mathbf{R}(\rho)\mid \rho\mbox{ a majority path from $y$ to $x$}\}\label{spliteq}
\end{equation}
by Lemma \ref{PathLem}, so $\mathsf{Advantage}(x,y,\mathbf{R}|_{\{x,y\}})>\mathsf{Standard}(x,y,\mathbf{R}^{-x,y})$ since the right-hand side of (\ref{spliteq}) depends only on $\mathbf{R}^{-x,y}$. If \textit{not} $xP(f(\mathbf{R}))y$ then either $\mathrm{Margin}_\mathbf{R}(y,x)\geq \mathrm{Margin}_\mathbf{R}(x,y)$ or $\mathrm{Margin}_\mathbf{R}(x,y)>\mathrm{Margin}_\mathbf{R}(y,x)$ but 
\begin{equation}\mathrm{Margin}_\mathbf{R}(x,y)\leq  \mbox{max}\{\mathrm{Split}\#_\mathbf{R}(\rho)\mid \rho\mbox{ a majority path from $y$ to $x$}\}.\label{spliteq2}\end{equation}
In the first case, it follows  that $\mathrm{Margin}_\mathbf{R}(x,y)\leq 0$ and so 
\[\mathsf{Advantage}(x,y,\mathbf{R}|_{\{x,y\}})\leq \mathsf{Standard}(x,y,\mathbf{R}^{-x,y}).\]
In the second case, since the right-hand side of (\ref{spliteq2}) depends only on $\mathbf{R}^{-x,y}$, it follows that $\mathsf{Advantage}(x,y,\mathbf{R}|_{\{x,y\}})\leq \mathsf{Standard}(x,y,\mathbf{R}^{-x,y})$.
\end{proof}

\subsection{Oligarchy under PI-weak IIA}

\OligarchyTheorem*
\begin{proof} Campbell and Kelly \citeyearpar{Campbell2000} show that if $f$ satisfies the hypothesis of the theorem, then (i) any coalition that is almost weakly decisive on some $x,y$ is weakly decisive\footnote{Recall that a coalition $C\subseteq V$ is  \textit{weakly decisive} (resp.~\textit{almost weakly decisive}) \textit{with respect to $f$} if for every profile $\mathbf{R}$ and $x,y\in X$, if $xP(\mathbf{R}_i)y$ for all $i\in C$ (resp.~and $yP(\mathbf{R}_j)x$ for all $j\in V\setminus C$), then $xf(\mathbf{R})y$.} (Lemmas 1 and 2 in \citealt{Campbell2000}, which do not use completeness of $f$), and (ii) the family of weakly decisive coalitions is closed under intersection (see the second paragraph of the proof of Lemma 3 in \citealt{Campbell2000}, which does not use completeness of $f$). The family of weakly decisive coalitions is also nonempty by Pareto. Let $C$ be the intersection of all weakly decisive coalitions, so $C$ is the smallest weakly decisive coalition. We show that each $i\in C$ is a vetoer. Let $x,y\in X$ and $\mathbf{R}$ be such that $xP(\mathbf{R}_i)y$. Assume toward a contradiction that $yP(f(\mathbf{R}))x$. Consider $\mathbf{R}'$ such that:
\[\mathbf{R}'|_{\{x,y\}}=\mathbf{R}|_{\{x,y\}};\]
\[xP(\mathbf{R}_i')aP(\mathbf{R}_i')y;\]
\[aP(\mathbf{R}_j')y\mbox{ and } aP(\mathbf{R}_j')x\mbox{ for all }j\in V\setminus \{i\}. \]
Then $aP(f(\mathbf{R}'))y$ by Pareto. If $aP(f(\mathbf{R}'))x$, then by PI-weak IIA we have $af(\mathbf{R}'')x$ for all profiles $\mathbf{R}''$ such that $aP(\mathbf{R}_j'')x$ for all $j\in V\setminus \{i\}$ and $xP(\mathbf{R}''_i)a$, i.e., $ V\setminus \{i\}$ is almost weakly decisive on $a,x$ and hence weakly decisive by (i), contradicting the choice of $C$. Thus, we have \textit{not} $aP(f(\mathbf{R}'))x$. So by transitivity, \textit{not} $yf(\mathbf{R}')x$. But this is a contradiction since from  $\mathbf{R}|_{\{x,y\}}=\mathbf{R}'|_{\{x,y\}}$ and $yP(f(\mathbf{R}))x$, we have $yf(\mathbf{R}')x$ by PI-weak IIA. So indeed \textit{not} $yP(f(\mathbf{R}))x$. Since $x,y,$ and $\mathbf{R}$ were arbitrary, $i$ is a vetoer. 
\end{proof}

\bibliographystyle{plainnat}
\bibliography{escape}

\end{document}